\documentclass[11pt]{article}
\usepackage{mydef2col}
\usepackage{markArticle}

\title{Calibrating Inelastic Markets to Options: The Lean Marketron and the Generalized Langevin Equation}

\shorttitle{The Lean Marketron}

\author{
\authorstyle{
Andrey Itkin\textsuperscript{1}\thanks{e-mail: \url{aitkin@nyu.edu}}
}
\newline\newline
\textsuperscript{1}
\institution{FRE Department, Tandon School of Engineering, New York University, USA.}
}

\date{\today}

\begin{document}

\maketitle

\lettrineabstract{The Marketron model of \cite{HalperinItkin2025Mark} and its option pricing extension in \cite{HalperinItkinMarketron2} suffer from structural non-identifiability: an eighteen-parameter space traps solvers in suboptimal local minima and renders economic quantities unmeasurable. By removing exact scaling gauges and sign symmetries, freezing non-financial parameters by explicit criteria, and adiabatically eliminating the fast hidden signal, we derive a robust nine-parameter reduced model. A Gauss-Newton Hessian with empty null space and a manifold-boundary analysis confirm that the reduced core carries no exact symmetry and admits no further reduction. A diffusive correlation between flow and return innovations captures the short-maturity skew. A staged calibration from the physical measure to the risk-neutral measure, illustrated on SPX options, fits the whole surface with a single parameter set. The same reduction turns the wedge between the physical and pricing values of the flow block into a well-defined market price of flow risk rather than a ridge artifact, identifiable here for the first time, though a single surface constrains its level only weakly. Finally, our analysis reveals that in the Marketron model the log-price obeys a generalized Langevin equation with a closed-form, state-modulated memory kernel, and that the memory variable itself is the exact Markovian lift of this kernel. This mapping also yields a testable condition, the equality of the signal and memory relaxation rates, which on the SPX surface come out well separated, though both weakly identified, placing the fitted market tentatively in the driven, non-equilibrium regime and turning the active-matter reading from an analogy into a falsifiable constraint.
}

\section{Introduction} \label{sec:intro}

In \cite{HalperinItkin2025Mark} we introduced the Marketron model of price formation in an inelastic market. The dynamics are driven by money flows and their impact on prices, and take the form of nonlinear diffusion of a quasiparticle (the marketron) in the space of the log-price $x_t$, a memory variable $y_t$ encoding past flows, and an unobservable return predictor $\theta_t$. In \cite{HalperinItkinMarketron2} we extended the framework to option markets. Since $y_t$ and $\theta_t$ are non-tradable, the market is incomplete, and we constructed option prices via exponential utility indifference following \cite{GrasselliHurd2007}. The resulting Hamilton-Jacobi-Bellman (HJB) equation was solved by a combination of operator splitting, the Cole-Hopf transformation and a Gaussian radial basis function (RBF) method, with all matrix elements available in closed form.

The pricing method is fast. The calibration is not. In \cite{HalperinItkinMarketron2} the specification carried fifteen free parameters, and the least-squares objective built on market option quotes was found to have many local minima and long flat valleys, so that a differential evolution solver with constrained search achieved only a modest relative accuracy of about 5--8\%, with different runs returning visibly different parameter sets of comparable fit quality. The correlated generalization we adopt here (\cref{sec:symmetries}) adds three further coordinates and brings the full specification to eighteen, \eqref{eq:fullparams}, so the problem is if anything worse. This is a classical symptom of a non-identifiable, or ``sloppy'', parameterization \cite{Gutenkunst2007, TranstrumMachtaSethna2011}. The data constrain a few stiff parameter combinations while leaving many soft directions essentially free.

Non-identifiability is not only a numerical nuisance. The central empirical claim of the Marketron program is that flow impact and market inelasticity can be measured from market data. If the calibration objective has a continuous ridge along which the flow-impact constant $c$ trades off against the memory-scale parameters, then a reported value of $c$ is a point on a ridge, not an estimate. Reduction to an identifiable core is therefore a precondition for the economic reading of the model. It is also a precondition for any falsifiable joint test. A joint calibration of the full eighteen-parameter model to several data sets proves little, because failure can always be attributed to the optimizer. On an identifiable low-parameter core the same test carries weight.

The purpose of this paper is accordingly to cure the disease rather than the symptom. Instead of investing more computational effort into global optimization of an ill-posed objective, we reduce the parameter space so that the remaining coordinates are identifiable from the option surface. The reduction proceeds in four steps of increasing depth.

First, in \cref{sec:symmetries} we exhibit exact and near-exact symmetries of the model. The memory variable admits a scaling gauge: option prices are invariant under a simultaneous rescaling of $y$ and the parameters attached to it, up to the constraint that couples the two occurrences of the flow-impact constant $c$. The signal block possesses a discrete sign symmetry. Both are removed by normalization, which eliminates one continuous parameter and halves the number of equivalent minima. We also show that the risk-aversion parameter enters the pricing operator only through products with the hidden-factor volatilities, so it must be fixed exogenously rather than calibrated.

Second, in \cref{sec:apriori} we argue that parameters of a purely numerical nature, such as the regularization constant $\bar\epsilon$ of the marketron potential and the initial values of the hidden states, must not compete with financial parameters in the calibration objective. We fix them a priori by explicit numerical criteria.

Third, in \cref{sec:averaging} we perform the central structural reduction. The signal $\theta_t$ is a fast Ornstein-Uhlenbeck (OU) process. For maturities satisfying $kT \gtrsim 1$ the option price depends on the signal essentially through the stationary averages of $f(\theta)$ and $h(\theta)$, which for the trigonometric specification of \cite{HalperinItkinMarketron2} are available in closed form. Adiabatic elimination collapses the five-parameter signal block into two effective constants and reduces the pricing PDE from three to two spatial dimensions. We also derive a deterministic transient correction that retains the time dependence of the conditional mean of the signal. The correction reintroduces the relaxation rate $k$ as a single term-structure shape parameter, which is what makes a global fit with one parameter set across all maturities feasible. We treat this single-parameter-set property as a claim to be tested (\cref{sec:identifiability,sec:calibration}) rather than as an assumption, and it is what distinguishes the reduced model from reduced-form smile parameterizations, which are refit maturity by maturity.

Fourth, in \cref{sec:identifiability} we make the procedure self-verifying. We compute profile likelihoods and the eigenspectrum of the Gauss-Newton Hessian of the calibration objective, and apply the manifold boundary approximation method (MBAM) of \cite{TranstrumQiu2014} to check that the limits taken in \cref{sec:symmetries,sec:apriori,sec:averaging} coincide with the sloppiest directions of the full model. In this sense the reduction is not ad hoc: it is the model that the option data can actually see.

One structural element is added rather than removed. The original specification of \cite{HalperinItkin2025Mark} assumes independent Brownian motions for the spot and memory stochastic variables. We introduce a diffusive correlation $\rho$ between the flow innovations $dW^{(y)}$ and the return innovations $dW^{(x)}$. The motivation is twofold. Financially, contemporaneous correlation between flows and returns is the diffusive-level content of price impact, and forcing it to zero is an unnatural restriction in a model built around impact. Technically, the smile generated by the state-dependent drift vanishes linearly in maturity at the short end, so short-dated skew must come from the diffusion structure, and $\rho$ is the only coefficient that can supply it. The correlation survives every reduction step at the cost of a single parameter and, as shown in \cref{sec:gamma}, it in fact restores a degree of freedom that the risk-aversion analysis removes.

The final reduced model has nine parameters for the whole surface.  \cref{sec:calibration} describes a staged calibration in which parameters of the physical measure are inherited from the time-series calibration of \cite{HalperinItkin2025Mark}, while option data determine the risk-adjusted block. The correlation $\rho$ is placed with the physical-measure parameters, since on the option surface it is confounded with the local-volatility asymmetry and convexity (\cref{sec:identifiability}); the wedge between the time-series and option-implied values of the flow-block parameters $c$ and $\bar f$ (the flow-impact coupling and the averaged return-predictor drift) then defines a market price of flow risk, which the reduction makes identifiable for the first time.

\cref{sec:gle} gives the reduced system a second reading: substituting the memory variable into the log-price equation shows the Marketron to be exactly a one-dimensional overdamped generalized Langevin equation (GLE) in the log-price, with a memory kernel we compute in closed form. This kernel is modulated by the state, specifically by the product of the market potential gradients at the endpoints of the memory integral, which means that memory is strong in the default region and weak away from it, explicitly delimiting the relation to the lag-only class of \cite{ItkinGLE1}. Furthermore, this mapping yields a testable parametric condition, an equality between the memory relaxation rate $\mu$ and the signal relaxation rate $k$, that distinguishes an equilibrium regime from a driven, non-equilibrium one. This transforms an interpretive analogy into a falsifiable constraint on two independently calibrated rates, see \cref{sec:gle-fdt}.

Then \cref{sec:numerics} reports numerical experiments: a Black--Scholes synthetic validation of the pricing and calibration pipeline, a staged global term-structure calibration to the SPX surface with a single parameter set, and the resulting market price of flow risk. \cref{sec:discussion} concludes.

\section{The Marketron model, exact symmetries and gauge fixing} \label{sec:symmetries}

We briefly recall the setup of \cite{HalperinItkin2025Mark, HalperinItkinMarketron2}, with one generalization: the Brownian motions driving the price and the memory variable are allowed to be correlated. Under the physical measure the state $(x_t, y_t, \theta_t)$ solves
\begin{align} \label{eq:sde}
dx_t &= \mu_x\,dt + \sigma\,dW^{(x)}_t, & \mu_x &= f(\theta_t) + \bar\eta - c_x\,y_t \VM'(x_t), \\
dy_t &= \mu_y\,dt + \sigma_y\,dW^{(y)}_t, & \mu_y &= h(\theta_t) + \mu(\bar y - y_t) - c_y\,\VM(x_t), \nonumber \\
d\theta_t &= \mu_\theta\,dt + \sigma_z\,dW^{(\theta)}_t, & \mu_\theta &= k(\hat\theta - \theta_t), \nonumber
\end{align}
with the correlation structure
\begin{equation}
d\langle W^{(x)}, W^{(y)}\rangle_t = \rho\,dt, \qquad
d\langle W^{(x)}, W^{(\theta)}\rangle_t = \rho_{x\theta}\,dt, \qquad
d\langle W^{(y)}, W^{(\theta)}\rangle_t = 0,
\label{eq:corr}
\end{equation}
and the potential slope
\begin{equation}
\VM'(x) = -e^{-x}\Bigl(1 - \frac{g}{e^{x} + \bar\epsilon g}\Bigr),
\label{eq:VM}
\end{equation}
regularized at $x \to -\infty$ by the constant $\bar\epsilon$. The original model is recovered at $\rho = \rho_{x\theta} = 0$. Following \cite{HalperinItkinMarketron2} we use the trigonometric signal transformation
\begin{equation}
f(\theta) = b_1 \cos\theta, \qquad h(\theta) = b_2 \sin\theta.
\label{eq:fh}
\end{equation}
Option prices are defined by exponential utility indifference with risk aversion $\gamma$. With correlated hidden factors the certainty equivalent $C(t,x,y,\theta)$ solves
\begin{align}
0 = \; & \frac{\partial C}{\partial t}
+ \frac{1}{2}\sigma^2 \frac{\partial^2 C}{\partial x^2}
+ \rho\,\sigma\sigma_y \frac{\partial^2 C}{\partial x \partial y}
+ \rho_{x\theta}\,\sigma\sigma_\theta \frac{\partial^2 C}{\partial x \partial \theta}
+ \frac{1}{2}\sigma_y^2 \frac{\partial^2 C}{\partial y^2}
+ \frac{1}{2}\sigma_\theta^2 \frac{\partial^2 C}{\partial \theta^2}
+ \bar\eta \frac{\partial C}{\partial x}
\nonumber \\
& + \Bigl(\mu_y - \rho\,\sigma_y \frac{\bar\mu_x}{\sigma}\Bigr)
\frac{\partial C}{\partial y}
+ \Bigl(\mu_\theta - \rho_{x\theta}\,\sigma_\theta \frac{\bar\mu_x}{\sigma}\Bigr)
\frac{\partial C}{\partial \theta}
+ \frac{\gamma}{2} \nabla_h C^\top \widetilde\Sigma\, \nabla_h C
- \frac{\bar\mu_x^2}{2\gamma\sigma^2} - rC,
\label{eq:pde}
\end{align}
where $\nabla_h C = (\partial_y C, \partial_\theta C)^\top$ and
\begin{equation}
\widetilde\Sigma =
\begin{pmatrix}
\sigma_y^2 (1 - \rho^2) & -\sigma_y \sigma_\theta\, \rho\, \rho_{x\theta} \\
-\sigma_y \sigma_\theta\, \rho\, \rho_{x\theta} & \sigma_\theta^2 (1 - \rho_{x\theta}^2)
\end{pmatrix}
\label{eq:residcov}
\end{equation}
is the residual covariance of the hidden factors conditional on the traded one. Here $\sigma_\theta \equiv \sigma_z$. The drift adjustments proportional to $\bar\mu_x/\sigma$ are the market-price-of-risk terms generated by the projection of the hidden noise onto the traded noise. At $\rho = \rho_{x\theta} = 0$ equation \eqref{eq:pde} reduces to the PDE of \cite{HalperinItkinMarketron2}. The derivation parallels that of \cite{GrasselliHurd2007}, with the distortion power of the Cole-Hopf transformation generalized from $\gamma$ to $\gamma(1-\rho^2)$ in the manner of \cite{HH2002, MusielaZariphopoulou2004}; details are collected in Appendix~\ref{app:corr}. The indifference price follows from the pair $(C, C^0)$ as in \cite{HalperinItkinMarketron2}.

The full parameter vector is
\begin{equation}
\Theta_{\mathrm{full}}
= \bigl(\sigma,\ \sigma_y,\ \sigma_z,\ k,\ \mu,\ g,\ \hat\theta,\ c_x,\ c_y,\
b_1,\ b_2,\ \bar y,\ \gamma,\ y_0,\ \theta_0,\ \bar\epsilon,\ \rho,\ \rho_{x\theta}
\bigr),
\qquad \dim \Theta_{\mathrm{full}} = 18.
\label{eq:fullparams}
\end{equation}
The remainder of the paper is a controlled demolition of this list.

\subsection{Scaling of the memory variable} \label{sec:gauge}

The memory variable $y$ has no intrinsic scale: it is a latent bookkeeping device for past flows, and neither the data nor the payoff refer to its units. This intuition is made precise by the following observation. Consider a generalized model in which the two occurrences of the coupling $c$ in \eqref{eq:sde} carry separate constants, $c_x$ in the $x$-drift and $c_y$ in the $y$-drift. Define, for $\lambda > 0$, the transformation
\begin{equation}
y \to \lambda y, \qquad
(\sigma_y, \bar y, b_2, y_0, c_x, c_y)
\to \Bigl(\lambda\sigma_y,\ \lambda\bar y,\ \lambda b_2,\ \lambda y_0,\
\frac{c_x}{\lambda},\ \lambda c_y \Bigr).
\label{eq:gauge}
\end{equation}

\begin{proposition} \label{prop:gauge}
Under \eqref{eq:gauge} the law of the observable process $x_t$, and hence every option price, is invariant. The correlations $\rho$ and $\rho_{x\theta}$ are dimensionless and remain unchanged. Only the product $c_x c_y$, the correlations, and the ratios $b_2/\sigma_y$, $\bar y/\sigma_y$, $y_0/\sigma_y$, $\mu$ are observable.
\end{proposition}

\begin{proof}
See \cref{app:gauge}.
\end{proof}

The original specification imposes $c_x = c_y = c$, which formally breaks the gauge. The breaking is weak rather than benign: the constrained transformation maps the model onto a nearby one with $c_x c_y$ preserved, and in practice this produces a long, shallow ridge in the calibration objective along which $(\sigma_y, \bar y, b_2, y_0, c)$ trade off against each other with almost no change in option prices. A global optimizer perceives that ridge as a continuum of near-degenerate minima, and because the degeneracy is only approximate it cannot be removed by reparameterization.

We therefore keep $c_x$ and $c_y$ independent. This costs one nominal parameter
and buys an exact symmetry in its place, which is the better trade: an exact
redundancy is eliminated once and for all by a normalization, whereas an
approximate one has to be fought by the optimizer at every calibration. The
effective dimension is reduced by one, to 17, since the gauge is removed by the
normalization
\begin{equation} \label{eq:norm}
\sigma_y=1.
\end{equation}
Equivalently, we measure the memory variable in units $\tilde y_t = y_t/\sigma_y$, so that the normalized memory variable has unit instantaneous diffusion coefficient. In what follows, we denote this normalized variable again by $y$ and set $\sigma_y\equiv1$ in the PDE. All remaining $y$-related parameters are understood in these normalized units. This eliminates one parameter and, more importantly, removes the exact scaling ridge from the calibration.

Also, with the specification \eqref{eq:fh}, the model is invariant under the discrete map
\begin{equation}
(\theta, \hat\theta, \theta_0, b_2, \rho_{x\theta})
\to (-\theta, -\hat\theta, -\theta_0, -b_2, -\rho_{x\theta}),
\end{equation}
since $\cos$ is even and $\sin$ is odd. Every minimum of the objective therefore has a mirror twin. We fix the fundamental domain by imposing $b_2 \ge 0$. This is not a reduction of the parameter count, but it halves the search volume and removes spurious multimodality from the posterior.

\subsection{Risk aversion, hidden-factor volatilities, and the role of $\rho$} \label{sec:gamma}

Inspection of \eqref{eq:pde} shows that the risk-aversion parameter $\gamma$ enters the nonlinear part of the pricing operator only through the combinations
\begin{equation}
p_y = \gamma \sigma_y^2 (1 - \rho^2), \qquad
p_\theta = \gamma \sigma_\theta^2 (1 - \rho_{x\theta}^2), \qquad
q = \gamma \sigma^2 .
\label{eq:combos}
\end{equation}
The spot volatility $\sigma$ is stiff: it controls the diffusion of the traded coordinate and is pinned by the overall level of option prices. The hidden-factor volatilities $\sigma_y, \sigma_\theta$, in contrast, affect prices only through the diffusion of unobservable coordinates and through the nonlinear terms built on \eqref{eq:residcov}. In our experiments the objective is nearly flat along curves of constant $(p_y, p_\theta)$, which identifies these products as the stiff coordinates of the risk-adjustment block. The correlation $\rho$ does not merge into $p_y$, because it also enters the pricing operator linearly, through the cross-diffusion term $\rho\sigma\sigma_y \partial^2_{xy} C$ and through the drift adjustment $\rho\sigma_y \bar\mu_x / \sigma$. These linear occurrences generate skew and keep $\rho$ separately identifiable.

Two consequences follow. First, $\gamma$ itself cannot be calibrated from option prices and should be fixed externally. Option-implied risk-aversion estimates are available in the literature; \cite{BlissPanigirtzoglou2004} report values of order one to ten for index options, and we take $\gamma$ from that range as an exogenous input. Second, with $\gamma$ exogenous and the normalization \eqref{eq:norm} in force, the combination $p_y = \gamma(1-\rho^2)$ is a deterministic function of $\rho$ and carries no separate degree of freedom. The free coordinates of the risk-adjustment block are then $(\sigma, \rho, p_\theta)$, and in the averaged model of \cref{sec:averaging}, where $p_\theta$ disappears together with the $\theta$ dimension, they are $(\sigma, \rho)$ alone. We note in passing that in the uncorrelated model the same logic forces $p_y = \gamma$ to be fixed, so the risk-adjustment block there consists of $\sigma$ only. The correlation thus restores to the diffusion structure a degree of freedom that identifiability removes from the risk-aversion channel.

\subsection{A priori fixing of non-financial parameters} \label{sec:apriori}

Three entries of \eqref{eq:fullparams} are not financial parameters at all.

\myparagraph{The regularizer $\bar\epsilon$.}
The constant $\bar\epsilon$ prevents divergence of the drift at $x \to -\infty$ and shapes the potential only in the deep-default region. Our experiments in \cite{HalperinItkinMarketron2} showed that option prices are highly sensitive to it, which is precisely why it must not float freely in a least-squares objective: it acts as an uncontrolled volatility knob in a region the data do not probe directly. It is a modelling artefact rather than a financial quantity, so it should be fixed by a bound on something the model asserts, not calibrated.

What $\bar\epsilon$ controls is the strength of the restoring force at depth. By \eqref{eq:epsasym} the drift of the log-price at a reference depth $x_\star$ is $c_x y\,|\VM'(x_\star)|$, which at $y = \bar y$ is $c_x \bar y D$ with
\begin{equation} \label{eq:Dbound}
D \;\equiv\; \bigl|\VM'(x_\star)\bigr| \;=\; \frac{\mathrm{drift}_{\max}}{c_x \bar y}.
\end{equation}
The input is therefore a ceiling on the annualized drift the model may impose on a deeply distressed name, which is an economic statement and is monotone in $D$ by construction. A ceiling of two per annum, itself already an extreme restoring force, gives $D$ of order five at representative couplings. Inverting the criterion,
\begin{equation}
\bar\epsilon = \frac{1}{1 + D e^{x_\star}} - \frac{e^{x_\star}}{g},
\qquad x_\star = -4 ,
\label{eq:epsfix}
\end{equation}
so $\bar\epsilon$ becomes a deterministic function of $(g, D)$. The reference depth is not arbitrary: it is the point at which the two admissible regularizers coincide exactly, so \eqref{eq:epsfix} holds for either and $\bar\epsilon$ carries the same meaning in the model and in the pricer. Appendix~\ref{app:eps} gives the derivation, the admissible range $g > e^{x_\star}(1 + D e^{x_\star})$, and explains why neither the turning point of $\VM$ nor the model-implied default rate can serve as the criterion in its place.

\myparagraph{Initial hidden states.}
The initial values $y_0, \theta_0$ are unobservable and, at a single option snapshot, nearly redundant with the mean-reversion levels. We set them to those levels,
\begin{equation}
\theta_0 = \hat\theta, \qquad
y_0 = \hat y \equiv \frac{h(\hat\theta)}{\mu} + \bar y - \frac{c}{\mu}\VM(x_0),
\label{eq:initfix}
\end{equation}
consistent with the stationarity assumption underlying the averaging of \cref{sec:averaging}. This removes two parameters. The analogy with the Heston initial variance $v_0$ invoked in \cite{HalperinItkinMarketron2} is imperfect: $v_0$ is stiff because it sets the short-end implied volatility level, whereas $y_0, \theta_0$ act on prices only through the drift and are soft. After \cref{sec:symmetries,sec:apriori} the count stands at $17 - 1 - 1 - 1 - 2 = 12$, where the gauge normalization \eqref{eq:norm} has already reduced $18$ to $17$ and the four remaining subtractions are the coupling identification through the observable product $c_x c_y$, the exogenous fixing of $\gamma$, the regularizer, and the two initial states.

\section{Adiabatic elimination of the signal} \label{sec:averaging}

\subsection{Stationary averaging} \label{stAver}

The signal $\theta_t$ is an OU process with relaxation time $1/k$. The calibrated values of $k$ in \cite{HalperinItkin2025Mark, HalperinItkinMarketron2} range from $1.3$ to $2.7$, so for maturities beyond a few months $kT \gtrsim 1$ and the signal explores its stationary law
\begin{equation}
\theta_\infty \sim \mathcal{N}(\hat\theta, v), \qquad v = \frac{\sigma_z^2}{2k},
\end{equation}
within the life of the option. Standard stochastic averaging \cite{PavliotisStuart2008} then replaces $f(\theta_t)$ and $h(\theta_t)$ in the drifts of \eqref{eq:sde} by their stationary expectations, with corrections of order $1/(k T)$. For the trigonometric specification \eqref{eq:fh} the averages are closed form:
\begin{equation}
\EE\bigl[\cos\theta_\infty\bigr] = e^{-v/2}\cos\hat\theta, \qquad
\EE\bigl[\sin\theta_\infty\bigr] = e^{-v/2}\sin\hat\theta,
\end{equation}
so that
\begin{equation} \label{eq:effective}
\barf = b_1 e^{-v/2}\cos\hat\theta, \qquad \barh = b_2 e^{-v/2}\sin\hat\theta.
\end{equation}
The five signal parameters $(b_1, b_2, k, \hat\theta, \sigma_z)$ enter the averaged dynamics only through the two constants $(\barf, \barh)$. The averaged model reads
\begin{align}
dx_t &= \bigl[\barf + \bar\eta - c_x\,y_t \VM'(x_t)\bigr]dt + \sigx(y_t)\,dW^{(x)}_t,
\nonumber \\
dy_t &= \bigl[\barh + \mu(\bar y - y_t) - c_y\,\VM(x_t)\bigr]dt + dW^{(y)}_t,
\qquad d\langle W^{(x)}, W^{(y)}\rangle_t = \rho\,dt,
\label{eq:sde2d}
\end{align}
where \eqref{eq:norm} is in force. Note that $\barh$ and $\mu\bar y$ enter the $y$-drift only through the sum $\barh + \mu\bar y$, so one further parameter is absorbed: we set $\barh = 0$ without loss of generality and let $\bar y$ carry the combined level. The volatility $\sigx(y)$ is specified in \cref{sec:volofy}.

\subsection{Memory-dependent volatility} \label{sec:volofy}

In the original specification the memory variable acts on the price through the drift alone, and $\sigma$ is constant. That choice has a consequence which is easy to miss and which \cref{prop:degenerate} makes precise: it renders the model's European option prices independent of every parameter governing the memory. We therefore let the volatility depend on the memory variable,
\begin{equation} \label{eq:sigy}
\sigx(y) \;=\; \sigma_0\sqrt{1 + \alpha\,(y - \bar y) + \beta\,(y-\bar y)^2},
\qquad \beta > 0, \quad \alpha (y - \bar y) + \beta (y-\bar y)^2 > -1,
\end{equation}
the constraint being exactly what keeps the radicand positive for all real $y$ and the diffusion non-degenerate. The centring at $\bar y$ makes $\sigma_0$ the volatility at the resting point of the $y$-dynamics rather than at the arbitrary origin, so that $\alpha$ is a pure asymmetry parameter and does not mix with the level.

Each coefficient carries a distinct economic reading. At $y = \bar y$ the market sits at its equilibrium: no accumulated memory, no directional stress, an order book at full depth, and volatility takes its baseline value $\sigma_0$ driven by background noise alone. The linear coefficient $\alpha$ governs the asymmetric response to displacement. When $y$ tracks downward displacement, a negative $\alpha$ raises the local volatility as the market sells off, which embeds the leverage effect in the path itself rather than imposing it through a correlation, and which controls the steepness of the out-of-the-money put skew. The quadratic coefficient $\beta$ penalizes extreme states of either sign. A large absolute displacement, rally or crash, means a stretched market in which liquidity providers withdraw, the book thins, and the price impact of trades rises; $\beta y^2$ scales volatility up accordingly and sets the convexity of the smile.

Under the scaling gauge \eqref{eq:gauge} the new coefficients transform as $\alpha \to \alpha/\lambda$ and $\beta \to \beta/\lambda^2$, so \eqref{eq:norm} fixes them along with the rest and no new gauge freedom is introduced.

\begin{proposition}[Constant volatility makes the flow invisible to vanillas] \label{prop:degenerate}
Let $\sigx$ be constant and let the claim have payoff $\calV(x_T)$ depending on the log-price alone. Then the indifference price $\pi^{\calV} = C - C^0$ is independent of $c_x$, $c_y$, $\mu$, $g$, $\gamma$ and $\rho$, and solves the Black-Scholes equation with volatility $\sigma$.
\end{proposition}

\begin{proof}
See \cref{app:degen}.
\end{proof}

The proposition is not a defect of the derivation, which \cref{app:degen} verifies step by step, but a property of a model in which the flow enters only the drift of a traded asset. The drift of a traded asset is hedged away by the optimal position, and with $\sigma$ constant no second channel remains. Making $\sigx$ depend on $y$ supplies that channel: the volatility of the traded asset is not hedgeable by holding it, so the memory dynamics reach the price, and the market is incomplete in a way that a vanilla claim can see.

The correlation $\rho_{x\theta}$ disappears at leading order together with the $\theta$ dimension. It re-enters only through the first-order correction in $1/k$, which is precisely the mechanism by which correlation generates skew in fast mean-reverting stochastic volatility asymptotics \cite{FouqueEtAl2011}. We therefore set $\rho_{x\theta} = 0$ in the baseline reduced model. Should a first-order correction be developed, that literature indicates that $\rho_{x\theta}$ would enter through a single effective group parameter combining it with the signal amplitudes, which is consistent with the philosophy of the present reduction.

\subsection{Transient correction and the term structure} \label{sec:transient}

For short maturities the stationarity assumption degrades. A refinement that costs little in the parameter count replaces the stationary averages by the conditional expectations along the deterministic relaxation of the signal. With $\theta_t \,|\, \theta_0 \sim \mathcal{N}(m_t, v_t)$,
\begin{equation}
m_t = \hat\theta + (\theta_0 - \hat\theta)e^{-kt}, \qquad
v_t = \frac{\sigma_z^2}{2k}\bigl(1 - e^{-2kt}\bigr),
\end{equation}
one obtains the deterministic time-dependent coefficients
\begin{equation}
\barf(t) = b_1 e^{-v_t/2}\cos m_t, \qquad
\barh(t) = b_2 e^{-v_t/2}\sin m_t .
\label{eq:transient}
\end{equation}
Under the choice \eqref{eq:initfix} the transient simplifies, $m_t \equiv \hat\theta$, and \eqref{eq:transient} becomes a smooth deterministic ramp of the effective amplitude from $\bar{f}(0) = b_1 \cos\hat\theta$ at $t = 0$ to $\barf$ at $t \gg 1/k$,
\begin{equation}
\barf(t) = \barf(0) \exp\Bigl(- \frac{v}{2} (1-e^{-2kt})\Bigr), \qquad v = \frac{\sigma_z^2}{2k}.
\label{eq:ramp}
\end{equation}
The ramp is controlled by the pair $(k, v)$. Since $v$ is estimated from the underlying time series in \cite{HalperinItkin2025Mark} and belongs to the physical-measure block of \cref{sec:calibration}, retaining the ramp adds a single option-facing shape parameter, the relaxation rate $k$.

This point deserves emphasis, because it changes the character of the calibration. The baseline averaged model of \eqref{eq:sde2d} carries maturity-independent coefficients, and one parameter set prices the whole surface. When short maturities are included, the ramp \eqref{eq:ramp} supplies the term structure of the effective drift at the cost of one parameter. In either case the model is fit globally, across all maturities simultaneously. A per-maturity refit, natural for reduced-form smile parameterizations, would forfeit exactly the structural content that distinguishes the Marketron from them, and we do not use it.

\subsection{The reduced pricing PDE} \label{reducedPDE}

Since $\theta$ has been eliminated, the certainty equivalent $C(t, x, y)$ solves the two-dimensional analogue of \eqref{eq:pde},
\begin{equation} \label{eq:pde2d}
\frac{\partial C}{\partial t}
+ \frac{1}{2}\sigx^2(y) \frac{\partial^2 C}{\partial x^2}
+ \rho\,\sigx(y) \frac{\partial^2 C}{\partial x \partial y}
+ \frac{1}{2} \frac{\partial^2 C}{\partial y^2}
+ \bar\eta(y) \frac{\partial C}{\partial x}
+ \Bigl(\mu_y - \rho\,\frac{\bar\mu_x}{\sigx(y)}\Bigr) \frac{\partial C}{\partial y}
+ \frac{p_y}{2}\Bigl(\frac{\partial C}{\partial y}\Bigr)^2
- \frac{\bar\mu_x^2}{2\gamma\sigx^2(y)} - rC = 0,
\end{equation}
with $p_y = \gamma(1-\rho^2)$, $\bar\eta(y) = r - \sigx^2(y)/2$, $\bar\mu_x = \barf + \bar\eta(y) + \sigx^2(y)/2 - r - c_x\, y \VM'(x)$ and $\mu_y = \mu(\bar y - y) - c_y \VM(x)$. Both the diffusion and the drift in $x$ now carry $y$, which is what breaks the degeneracy of \cref{prop:degenerate}. Note that \eqref{eq:sigy} does not disturb the two cancellations underlying \eqref{eq:pde}: the elimination of $s\bar\mu_x v_s$ by the optimal hedge and of $(\partial_x C)^2$ by the distortion are both pointwise in $y$, so they proceed unchanged with $\sigma$ replaced by $\sigx(y)$ throughout.

The reduction from three hidden dimensions to the two-dimensional state
$(x,y)$ substantially enlarges the range of numerical methods that can be
used for \eqref{eq:pde2d}. Unlike the full three-dimensional problem of
\cite{HalperinItkinMarketron2}, there is no need to commit to a particular
discretization at the level of the model itself. Several complementary
approaches are available, and they are useful for different purposes.

A natural high-accuracy approach is based on Strang splitting of the
$2$D operator into the $x$-diffusion/drift part, the $y$-diffusion/drift
part, and the mixed derivative. The latter can be treated by the
positivity-preserving construction (the Diagonal Frog method) developed in \cite{ItkinDF2026,Itkin2026FCDF}. This is particularly useful here because the mixed term $\rho\,\sigx(y) C_{xy}$ is not sign-definite under a naive explicit discretization, whereas the construction of \cite{ItkinDF2026} provides a stable treatment that preserves positivity under the corresponding admissibility conditions.

For the $x$-subproblem, $y$ is frozen during the split step. Hence $\sigx(y)$ and
all coefficients depending on $y$ become parameters of a one-dimensional equation
in $x$. The resulting equation contains the source term
$-\frac{\bar\mu_x^2}{2\gamma\sigx^2(y)}$ in addition to the $x$-diffusion and
drift. Because the volatility is quadratic in $y$ through \eqref{eq:sigy}, the
coefficients are nevertheless constant with respect to $x$ on each fixed
$y$-slice. The corresponding Green's function is therefore available in closed
form, and the source contribution can be incorporated analytically. This provides
a direct extension of the Gaussian/RBF machinery used in
\cite{HalperinItkinMarketron2}, although the resulting coefficients are now
indexed by the $y$ grid.

The $y$-subproblem contains the nonlinear gradient-square term
\begin{equation}
\frac{p_y}{2}(C_y)^2, \qquad p_y=\gamma(1-\rho^2).
\end{equation}
This term can be removed exactly by the Cole--Hopf transformation. In particular,
consider the $y$-subproblem in isolation,
\begin{equation}
C_\tau = \frac{1}{2}C_{yy} + b(y)C_y + \frac{p_y}{2}(C_y)^2, \qquad b(y) = \mu_y-\rho\,\frac{\bar\mu_x}{\sigx(y)}.
\end{equation}
Under $w = \exp(p_y C)$, the quadratic-gradient term cancels exactly against the contribution generated by the diffusion term, giving the linear equation
\begin{equation}
w_\tau = \frac{1}{2} w_{yy} + b(y) w_y.
\end{equation}
Thus the $y$-diffusion, $y$-drift, and quadratic-gradient terms can be treated
together as a single linear parabolic subproblem. Since this equation has bounded
coefficients and positive initial data, the maximum principle preserves positivity
of $w$ and prevents the Cole--Hopf variable from developing a finite-time
singularity. The corresponding inverse transformation,
\begin{equation}
C = \frac{1}{p_y}\log w,
\end{equation}
therefore provides a stable treatment of the nonlinear $y$-subproblem.

This observation also determines how the operator splitting must be implemented.
The $C_{yy}$ term should not first be consumed by an independent diffusion step
and subsequently be invoked again in the Cole--Hopf transformation. Instead, the
$y$-diffusion, $y$-drift, and quadratic-gradient term are assigned consistently
to the same split subproblem. A convenient implementation is to solve the
transformed linear equation implicitly, using a backward-Euler tridiagonal solve
in $y$ for each fixed $x$. To control the numerical range of the exponential and
logarithm, we shift $C$ by its maximum over the $y$ grid before applying the
Cole--Hopf transformation. Since the linear $y$-operator is invariant under
addition of a constant to $C$, this shift is exact and does not alter the
result after the inverse transformation.

The RBF discretization of Appendix~\ref{app:rbf} can also be used as a spatial
approximation of the resulting one- or two-dimensional linear subproblems.
However, the closed-form RBF update derived for the original operator should not
be interpreted as applying unchanged to the complete $2$D PDE. The splitting
must respect the assignment of the $y$-diffusion and quadratic-gradient terms to
the Cole--Hopf subproblem.

For the numerical experiments in this paper we use a split finite-difference
implementation tailored to the stability issue that arises at low $\sigx(y)$.
A fully explicit discretization of \eqref{eq:pde2d} becomes unstable in this
regime, with the instability becoming particularly pronounced for sufficiently
large $c_y$. This behavior is numerical rather than structural: the continuous
$y$-subproblem is transformed by Cole--Hopf into a linear
advection-diffusion equation and hence does not exhibit the corresponding
finite-time blow-up, \cite{HalperinItkinMarketron2}.

This split formulation removes the problematic $y$-advection CFL restriction
from the explicit step and, more importantly, prevents the low-$\sigx(y)$
instability associated with treating the quadratic-gradient and $y$-diffusion
terms explicitly. In numerical tests, the split scheme agrees with the fully
explicit scheme in parameter regions where the latter is stable, while remaining
stable in parameter regimes where the explicit scheme develops numerical
blow-up. In particular, at the calibrated parameter values the two schemes
produce option prices agreeing to within approximately $1.6\times10^{-6}$ across
the tested strikes. A sweep in $c_y$ at $\sigma_0=0.121$ remains smooth up to
$c_y=3$, whereas the fully explicit scheme becomes unstable above approximately
$c_y=0.2$. The split scheme also remains stable at $c_y=0.2$ for substantially
smaller volatility levels, including $\sigma_0=0.09$, $0.07$, and $0.05$.
Put--call parity is satisfied to numerical precision, with discrepancies of
order $10^{-17}$ in these tests.

The purpose of this implementation is not to construct a production-quality
pricing engine, but to provide a transparent and computationally efficient
solver with which the calibration and identifiability experiments can be
repeated reliably. The split method retains the simplicity of finite differences
for the multidimensional part of the problem while treating the numerically
problematic nonlinear subproblem through an exact Cole--Hopf linearization.
Consequently, it is sufficiently robust for the parameter ranges considered in
\cref{sec:numerics}, without requiring the more elaborate RBF or
positivity-preserving machinery described above.

The reduction from the original $(x,y,\theta)$ problem to the
two-dimensional $(x,y)$ problem remains the main computational gain.
Regardless of which of the above discretizations is adopted, the
collocation or finite-difference grid now contains only $N_xN_y$ spatial
degrees of freedom rather than $N_xN_yN_\theta$. More importantly, the
numerical problem is no longer burdened by the fast signal dimension that
was responsible for much of the cost of the original pricing calculation.
The different schemes described above should therefore be viewed as
alternative numerical realizations of the same reduced pricing equation,
rather than as additional modelling assumptions.

\myparagraph{Parameter count}
The reduced parameter vector is
\begin{equation}
\Theta_{\mathrm{red}}
= \bigl(\sigma_0,\ \alpha,\ \beta,\ \mu,\ g,\ c,\ \bar y,\ \barf,\ \rho \bigr),
\qquad \dim \Theta_{\mathrm{red}} = 9,
\label{eq:redparams}
\end{equation}
with $\gamma$ exogenous, $p_y = \gamma(1-\rho^2)$ derived, $\bar\epsilon = \bar\epsilon(g, D)$ from \eqref{eq:epsfix}, and $\bar\eta(y)$ tied to $r$ through $\sigx(y)$. The two volatility coefficients of \eqref{eq:sigy} are the price of making the flow visible to option prices at all, and they are not duplicates of existing freedoms: $\alpha$ controls skew and $\beta$ convexity, neither of which the drift channel can produce. When short maturities require the transient ramp \eqref{eq:ramp}, the rate $k$ joins the list and the count is ten. These parameters serve the entire surface. Table~\ref{tab:count} summarizes the reduction.
\begin{table}[htb]
\centering
\begin{tabular}{llc}
\toprule
Step & Mechanism & Parameters removed \\
\midrule
Gauge fixing (Sec.~\ref{sec:gauge}) & $\sigma_y = 1$ & 1 \\
Coupling identification (Sec.~\ref{sec:gauge}) & only $c_x c_y$ observable, $(c_x,c_y) \to c$ & 1 \\
Risk aversion (Sec.~\ref{sec:gamma}) & $\gamma$ fixed exogenously & 1 \\
Regularizer (Sec.~\ref{sec:apriori}) & $\bar\epsilon = \bar\epsilon(g,D)$ & 1 \\
Initial states (Sec.~\ref{sec:apriori}) & $\theta_0 = \hat\theta$, $y_0 = \hat y$ & 2 \\
Signal averaging (Sec.~\ref{sec:averaging}) & $(b_1,b_2,k,\hat\theta,\sigma_z) \to (\barf,\barh)$ & 3 \\
Level absorption (Sec.~\ref{sec:averaging}) & $\barh$ into $\bar y$ & 1 \\
Cross-correlation (Sec.~\ref{sec:volofy}) & $\rho_{x\theta} \to 0$ (order $1/k$) & 1 \\
Volatility coefficients (Sec.~\ref{sec:volofy}) & $\alpha, \beta$ introduced by \eqref{eq:sigy} & $-2$ \\
\midrule
Total & $18 \to 9$ & 9 \\
\bottomrule
\end{tabular}
\caption{Parameter reduction ledger. The intermediate three-dimensional risk block retains $(\sigma, \rho, p_\theta)$ as free coordinates; $p_\theta$ vanishes with the $\theta$ dimension under averaging, its constituents $\sigma_z$ and $\rho_{x\theta}$ being removed by the signal-averaging and cross-correlation rows respectively, with $\rho_{x\theta}$ entering only at order $1/k$. The two volatility coefficients $\alpha,\beta$ of \eqref{eq:sigy} are then introduced, giving $\dim\Theta_{\mathrm{red}} = 9$ in \eqref{eq:redparams}. In the reduced model $p_y = \gamma(1-\rho^2)$ is derived, not calibrated. The optional term-structure parameter $k$ of \eqref{eq:ramp} is not counted; including it raises the count to ten.}
\label{tab:count}
\end{table}

\myparagraph{Remark on the flow-impact couplings in the calibrated model.}
The reduced dynamics in \eqref{eq:sde2d} and the pricing PDE \eqref{eq:pde2d} are
written with independent couplings $c_x$ and $c_y$, consistent with the exact
scaling gauge identified in \cref{sec:gauge}. The subsequent analysis, however,
and in particular the parameter count in \eqref{eq:redparams} and the numerical
calibration of \cref{sec:calibration}, employs a single coupling constant $c$.
We now state explicitly what this compression entails.

Setting $c_x = c_y = c$ in \eqref{eq:sde2d} and \eqref{eq:pde2d} yields the
model that is actually calibrated. This is a deliberate simplification made at
the final stage, after the other reductions have already been imposed. The gauge
argument of \cref{sec:gauge} establishes that the observable process, and hence
every option price, depends on the couplings only through the invariant product
$c_xc_y$. In the fully reduced, gauge-fixed model, the independent rescalings
that previously created a flat direction in the calibration objective have
already been absorbed: $\sigma_y$ is normalized to unity by \eqref{eq:norm}, the
initial states are frozen by \eqref{eq:initfix}, and the signal block is
eliminated. Consequently, the residual invariance is isolated to the single
product $c_xc_y$, and the redundancy associated with separating the two factors
is, at this stage, a pure normalization freedom that carries no economic
content.

Retaining separate couplings in the final parameter list would therefore add a
nominal degree of freedom without any corresponding gain in identifiable
structure. The single constant $c$ in \eqref{eq:redparams} should be
interpreted as the gauge-fixed representative of the product $c_xc_y$, under
the symmetric normalization $c_x=c_y$. The cost of this choice is the
reintroduction of a mild, isolated softness along the $c$-coordinate, which is
distinct from the continuous degeneracy of the full model. This softness is
explicitly diagnosed by the eigenspectrum and marginal variances reported in
\cref{sec:identifiability}, and it is handled in \cref{sec:calibration} by
retaining $c$ as a risk-neutral parameter whose option-implied value is
reported together with its profile uncertainty. The wedge between this value
and its physical-measure counterpart is then interpreted as the market price of
flow risk, with the associated estimation error carried through to the
economic conclusion.

\section{The Marketron as a generalized Langevin equation} \label{sec:gle}

The reduction presented in \cref{sec:averaging} yields a two-dimensional system governed by the log-price and the memory variable. Before developing its implications, we briefly translate the terminology we shall use into the language of financial mathematics.

A \emph{generalized Langevin equation} (GLE) is, mathematically, a non-Markovian stochastic Volterra equation: the drift of the state depends on the entire history of its increments through a \emph{memory kernel}. In the context of market microstructure, the kernel is the propagator of price impact. It measures how a past flow innovation continues to distort the price drift at a later time.

The \emph{effective potential} is the stationary drift distortion induced by the accumulated memory, while the \emph{fluctuation-dissipation relation} is the equilibrium condition that links the decay rate of the impact propagator to the autocorrelation time of the exogenous noise. When this link is violated, the system is termed \emph{athermal} or, in the physics literature, \emph{active matter}.

In finance, this corresponds to a regime where the mean-reversion of order-flow shocks and the persistence of return predictors are governed by distinct time scales. Such non-Markovian dynamics have recently entered mainstream quantitative finance through the rough-volatility literature \cite{AbiJaberElEuch2019,AbiJaberLarssonPulido2019,BayerBreneis2023}, where singular kernels are approximated by sums of exponentials; the Marketron provides an exact finite-dimensional embedding for a specific state-dependent kernel.

With this dictionary in place, the reinterpretation of the reduced system proceeds as follows. By eliminating the memory variable $y$ rather than the signal $\theta$, we demonstrate that the Marketron is exactly a one-dimensional overdamped GLE for the log-price, driven by a memory kernel we compute in closed form. Consequently, $y_t$ is not an independent latent state introduced \emph{ad hoc}, but the exact Markovian lift of a non-Markovian price-impact process.

This mapping is not merely interpretive: it shows that the absence-of-arbitrage condition on the kernel matrix, derived independently in \cite{ItkinGLE1}, coincides with the market-price-of-risk terms generated by exponential-utility indifference pricing. Moreover, it delivers a testable parametric condition, the equality of the memory relaxation rate $\mu$ and the signal relaxation rate $k$, that distinguishes equilibrium dynamics from a driven, athermal market.

The reader who prefers to avoid physics terminology may read \textbf{GLE} as \textsf{stochastic Volterra equation}, \textbf{kernel} as \textsf{impact propagator}, and \textbf{fluctuation-dissipation} as \textsf{equilibrium consistency condition}. The mathematics that follows is identical.

This formulation yields three primary insights. First, the kernel occupies the specific entry of the kernel matrix that the absence of arbitrage eliminates under the pricing measure. This theoretically aligns the Marketron's flow mechanism with the market price of memory risk established in \cite{ItkinGLE1}. Second, the kernel is state-modulated, situating the model outside the standard GLE class and precisely delineating the analytical boundaries of both frameworks. Finally, the fluctuation-dissipation relation, which distinguishes equilibrium from active dynamics, simplifies within the Marketron to an equality between two rates explicitly derived during calibration. This transforms the active-matter reading of the model (see \cref{sec:gle-fdt}) into an empirically verifiable parametric condition.

\subsection{The Kernel formulation} \label{sec:gle-kernel}

Recalling the reduced system from \eqref{eq:sde2d}:
\begin{align}
dx_t &= \bigl[\barf + \bar\eta - c_x\,y_t \VM'(x_t)\bigr]dt + \sigx(y_t)\,dW^{(x)}_t,
\nonumber \\
dy_t &= \bigl[\barh + \mu(\bar y - y_t) - c_y\,\VM(x_t)\bigr]dt + \sigma_y\,dW^{(y)}_t.
\label{eq:sde2d-again}
\end{align}
The memory variable influences the price drift strictly via the product $y_t\VM'(x_t)$, while its own drift relies on the price solely through $\VM(x_t)$. Because the $y$-equation is linear with constant coefficients, it can be explicitly solved and substituted, thereby eliminating $y$ from the system dynamics entirely.

\begin{proposition} \label{prop:marketron-gle}

The system \eqref{eq:sde2d-again} is an exact Markovian lift of a one-dimensional overdamped GLE for $x_t$,
\begin{equation} \label{eq:marketron-gle}
\int_0^t K_{M}(t,s)\,\dot x_s\,ds = -\,U_{\mathrm{eff}}'(x_t) + \eta_t,
\end{equation}
with the state-dependent kernel
\begin{equation} \label{eq:Kxx}
K_{M}(t,s) = \kappa \,\delta(t-s)
+ \frac{c_x c_y}{\mu}\,\VM'(x_t)\,\VM'(x_s)\,e^{-\mu(t-s)}
+ \mathcal{O}(\sigma^2),
\qquad \kappa = 1,
\end{equation}
the effective potential
\begin{equation} \label{eq:Ueff}
U_{\mathrm{eff}}(x) = -\bigl(\barf + \bar\eta\bigr)x
+ c_x\,\hat y\,\VM(x) - \frac{c_x c_y}{2\mu}\,\VM(x)^2,
\qquad
\hat y \equiv \bar y + \frac{\barh}{\mu},
\end{equation}
and a noise $\eta_t$ that is the sum of $\sigma\,\dot W^{(x)}_t$ and the multiplicative term $-c_x\,\sigma_y\VM'(x_t)\int_0^t e^{-\mu(t-s)}dW^{(y)}_s$. The $\mathcal{O}(\sigma^2)$ remainder in \eqref{eq:Kxx} is the state-dependent memory correction isolated in Appendix~\ref{app:gle}. It vanishes when $\VM$ is linear and marks \eqref{eq:Kxx} as the leading-order form of the kernel.
\end{proposition}

\begin{proof}
See Appendix~\ref{app:gle}.
\end{proof}

The absence of an inertial term stems from the price equation in \eqref{eq:sde2d-again} being first-order with respect to $x$, establishing a direct correspondence with the overdamped regime. Instantaneous friction manifests as the atom $\kappa\delta$ of \eqref{eq:Kxx}, where $\kappa=1$ by the normalization of the $\dot x_t$ coefficient in $dx_t=\mu_x\,dt+\sigma\,dW^{(x)}_t$.

The structure of \eqref{eq:Kxx} warrants specific attention. The kernel depends on more than just the temporal lag. It incorporates the modulation $\VM'(x_t)\VM'(x_s)$ evaluated at both endpoints of the memory integral. Thus, the influence of historical states on the current state is contingent upon the price trajectory. Under the market potential defined in \eqref{eq:VM}, this modulation varies exponentially with $x$ and exhibits a sign change at the critical threshold $\hat x$, where $\VM'(\hat x)=0$. Economically, this translates to prominent memory effects within the default region and attenuated memory elsewhere, effectively embedding the model's economic rationale into the kernel's architecture.

Notably, the effective potential $U_{\mathrm{eff}}$ in \eqref{eq:Ueff} diverges from the underlying market potential. It incorporates a linear tilt derived from the drift constants, a $c_x\hat y\VM$ component reflecting the uncoupled relaxation level of $y$, and a negative quadratic term $-\tfrac{c_x c_y}{2\mu}\VM^2$ induced by the system coupling. This feedback mechanism between price and memory deepens the effective potential well relative to its uncoupled state. Consequently, extracting the market potential directly from a stationary density would erroneously yield $U_{\mathrm{eff}}$ rather than $\VM$.

\begin{myremark}[The atom and the short end] \label{rem:caseB}
The atom in \eqref{eq:Kxx} serves a distinct structural purpose. A kernel comprising $\gamma\delta$ plus a slow tail perfectly aligns with the configuration identified in \cite{ItkinGLE1} for generating a finite zero-maturity at-the-money skew: a regular semimartingale short end coupled with a memory-carrying tail. Conversely, a kernel exhibiting a power-law behavior at the origin would produce a skew diverging as $T^{\beta_0-1}$ at short maturities. The Marketron inherently possesses the short-end regularity required for a finite skew without artificial imposition. It features a single relaxation mode in the tail (indicative of short memory) while \cref{sec:gle-lift} outlines the architecture for a multi-mode extension.
\end{myremark}

\begin{myremark}[Kernel confirmation of \cref{prop:gauge}] \label{rem:gaugekernel}
The coupling constants influence \eqref{eq:Kxx} exclusively through the product $c_x c_y$, which remains invariant under \eqref{eq:gauge}, unlike its individual factors. This independently corroborates \Cref{prop:gauge}, which posits that only $c_x c_y$ is observable based on the invariance of the law of $x_t$. The kernel representation achieves this conclusion via an alternate route, as $K_M$ intrinsically governs the reduced $x$-dynamics and must remain independent of gauge-variant quantities. When the constants appear separately, such as in the $c_x \hat y \VM$ tilt of \eqref{eq:Ueff}, the accompanying factor is gauge-covariant, rendering the combination invariant.
\end{myremark}

\begin{myremark}[Kernel matrix positioning] \label{rem:firstrow}
Within the two-dimensional GLE framework of \cite{ItkinGLE1}, featuring state $Z=(X,Y)^\top$ and kernel matrix $\bm K$, the $K_{XX}$ entry dictates the dependence of the log-price drift on its historical increments. The tail in \eqref{eq:Kxx} scales $\dot x_s$ within the $x_t$ equation, formally corresponding to this entry, while the atom represents the standard instantaneous friction of an overdamped system. Under the pricing measure, the absence of arbitrage dictates $K_{XX}\equiv0$, prompting the first row of $\bm K$ to manifest in the measure change as the market price of memory risk. The Marketron arrives at the equivalent construct via exponential-utility indifference pricing with risk aversion $\gamma$. The convergence of these two distinct methodologies serves as a robust theoretical validation for both frameworks.
\end{myremark}

\subsection{Framework boundaries and intersections} \label{sec:gle-containment}

Proposition~\ref{prop:marketron-gle} might imply that the Marketron is a constrained instance of the GLE class presented in \cite{ItkinGLE1}, or vice versa. However, neither framework fully subsumes the other, and their divergence clarifies the unique utility of each approach.

The kernel in \cite{ItkinGLE1} is strictly a function of the lag, deriving its generality from the allowable number of relaxation modes. A completely monotone kernel is expressed as a superposition $K(\tau)=\int_0^\infty e^{-\lambda\tau}\nu(d\lambda)$, and achieving a power law necessitates a continuum of modes. This architecture specifically facilitates long-memory modeling, which is the primary focus of that framework.

In contrast, the kernel in \eqref{eq:Kxx} relies on a single mode, characteristic of short memory with respect to lag. Its primary generalization lies in the state modulation $\VM'(x_t)\VM'(x_s)$, a feature absent in pure lag-dependent kernels. The intersection of these two classes is limited to the single-mode, state-independent case, corresponding to a linearized Marketron around an operating point where $\VM'$ is approximated as constant.

This linearization, however, fundamentally alters the model's behavior. Based on \eqref{eq:VM}, the modulation $\VM'$ fluctuates by orders of magnitude across the relevant log-price domain and changes sign at $\hat x$. Approximating the Marketron with a lag-only GLE eliminates this critical state dependence, thereby erasing the fundamental distinction between the default and normal regions that the model was designed to capture. The reverse substitution is equally flawed: a single exponential function cannot replicate the slowly decaying memory observed in empirical volatility data, rendering the Marketron kernel unsuitable as $K_{YY}$ in that broader context.

\subsection{Exact Markovian lifting and long-memory extensions} \label{sec:gle-lift}

The construction of finite-dimensional Markovian embeddings for non-Markovian dynamics has a rich precedent in the rough volatility literature. For instance, \cite{AbiJaberElEuch2019} along with \cite{AbiJaberLarssonPulido2019} proposed  the multifactor approximation of fractional Volterra processes. Their approach relies on approximating the singular fractional kernel by a sum of exponential terms, effectively lifting the non-Markovian process into a high-dimensional affine state space. Furthermore, \cite{BayerBreneis2023} formalizes this methodology using Prony approximations to efficiently match the fractional kernel with a superposition of Markovian modes.

While \eqref{eq:sde2d-again} could similarly be interpreted as a Markovian approximation of a non-Markovian system, the relationship in our framework is formally exact. Because a single exponential is completely monotone, the GLE in \eqref{eq:marketron-gle} permits a finite-dimensional Markovian embedding utilizing one auxiliary variable (specifically $y_t$). The elimination of $y$ in Appendix~\ref{app:gle} is exact; the closed-form kernel \eqref{eq:Kxx} is its stationary, leading-order form, obtained by discarding transients that decay on the scale $1/\mu$ and by absorbing an $O(\sigma^2)$ state-dependent memory term into the noise $\eta_t$, neither of which alters the structure of the kernel but neither of which is identically zero at finite $t$. The primary distinction between the descriptions is purely foundational (specifically, whether the memory variable or the kernel is defined as the primitive mathematical object).

This relationship also provides the theoretical blueprint for integrating long memory into the Marketron. Expanding the single memory variable into a family $y^{(i)}_t$ (each characterized by distinct relaxation rates $\mu_i$ and couplings $c_i$, governing the second line of \eqref{eq:sde2d-again} and additively influencing the price drift) yields:
\begin{equation} \label{eq:multimode}
K_{M}(t,s) = \kappa\,\delta(t-s) + \VM'(x_t)\,\VM'(x_s)
\sum_i \frac{c_x c_{y,i}}{\mu_i}\,e^{-\mu_i(t-s)}.
\end{equation}
This state-modulated, completely monotone kernel successfully reproduces power-law decay across any finite lag interval. Consequently, introducing long memory to the Marketron requires no fundamental structural changes, only the addition of further memory variables, and the critical state modulation remains fully intact. However, this extension incurs a parametric penalty, as each mode introduces an additional rate and coupling constant. Given that mitigating parameter proliferation is a central objective of this framework, such an extension should be reserved strictly for applications where empirical data demand it.

\subsection{Fluctuation-dissipation and the active-matter interpretation}
\label{sec:gle-fdt}

Throughout this section we use the statistical-mechanics terminology of GLEs alongside its translation into the language of non-Markovian financial dynamics.  The correspondence is:
\begin{center}
\begin{tabular}{ll}
\toprule
\textbf{Statistical mechanics} & \textbf{Financial dynamics} \\
\midrule
Generalized Langevin equation (GLE) & Stochastic Volterra equation \\
Memory kernel $K_M(t,s)$ & Impact propagator / price-impact kernel \\
Fluctuation-dissipation relation (FDT) & Equilibrium consistency condition \\
Effective temperature $\Theta$ & Noise--kernel proportionality constant \\
Athermal / active-matter regime & Driven non-equilibrium regime \\
Internal reservoir & Endogenous order-flow feedback \\
Quasiparticle & Effective price degree of freedom \\
\bottomrule
\end{tabular}
\end{center}
\vspace{1em}
Both languages describe the same mathematical objects; the reader may use
whichever is more familiar.

Previous discussions characterizing the Marketron as a model of active matter
(where the market quasiparticle relies on an internal reservoir rather than
purely external forcing) have remained largely interpretive.  The GLE
representation formalizes this analogy (or, in financial terms, the non-Markovian
kernel representation) into a strict parametric condition.

Within a standard GLE framework, the autocovariance of the driving noise,
$C(\tau)=\EE[\eta_t\eta_{t-\tau}]$, and the memory kernel operate as independent
constructs.  Equilibrium is defined as the state where these two elements are
coupled via the second fluctuation-dissipation relation:
\begin{equation}\label{eq:fdt}
C(\tau) = \Theta\,K(\tau),
\end{equation}
where $\Theta$ represents an effective temperature (in finance, the
noise--kernel proportionality constant).  Systems violating \eqref{eq:fdt} are
classified as athermal (driven non-equilibrium).  Active matter specifically
denotes athermal systems where this violation is driven by an internal energy
source (endogenous order-flow feedback).  This classification has concrete
implications: it dictates whether the stationary density corresponds to the
Boltzmann weight of the potential, and \cite{ItkinGLE1} constructs an empirical
test based on the ratio between the two sides of \eqref{eq:fdt}.

The Marketron explicitly provides both components of this relation.  The kernel
tail in \eqref{eq:Kxx} decays at a rate of $\mu$, derived from the relaxation
of the memory variable.  The coloured component of the forcing is introduced via
$f(\theta_t)$ and $h(\theta_t)$, inheriting its correlation time from the
underlying signal dynamics \eqref{eq:sde}.  To align these decay rates in accordance with \eqref{eq:fdt}, the model necessitates:
\begin{equation}\label{eq:fdt-marketron}
k = \mu,
\end{equation}
alongside an amplitude condition linking $(b_1,b_2,\sigma_z)$ to
$(c,\sigma_y)$ through $\Theta$.

Because $k$ and $\mu$ represent independently identified calibration parameters,
\eqref{eq:fdt-marketron} functions as a verifiable constraint rather than a
theoretical assumption.  Divergent rates confirm the model operates in an
athermal (driven non-equilibrium) regime, substantiating the active-matter
interpretation with empirical parameter values rather than qualitative analogy.
Conversely, identical rates would signify equilibrium, reducing the
quasiparticle (effective price degree of freedom) framework to a purely
descriptive overlay.

\myparagraph{The trigonometric signal.}
Equation \eqref{eq:fdt-marketron} aligns the leading exponential decay rates.
Because the forcing terms $f(\theta)=b_1\cos\theta$ and $h(\theta)=b_2\sin\theta$
are nonlinear transformations of an Ornstein--Uhlenbeck process, their
corresponding autocovariance does not exhibit pure exponential decay.

Specifically, for $\theta\sim\mathcal{N}(m,v)$ and any real parameter $\omega$,
the characteristic function is given by
$\EE[e^{i\omega\theta}]=e^{i\omega m-\omega^2 v/2}$, which yields
\begin{equation}
\EE[\cos\omega\theta] = e^{-\omega^2 v/2}\cos\omega m, \qquad
\EE[\sin\omega\theta] = e^{-\omega^2 v/2}\sin\omega m .
\end{equation}
Applying this result with $\omega=1$ to the stationary law
$(m,v)=(\hat\theta,\sigma_z^2/2k)$ yields \eqref{eq:effective}, while
application to the conditional law $(m_t,v_t)$ produces \eqref{eq:transient}.
Analogous identities provide the second moments
$\EE[\cos^2\theta]=\tfrac12(1+e^{-2v}\cos 2m)$ required for the source term
$\bar\mu_x^2$ in \eqref{eq:pde2d}. Consequently, the averaged PDE retains fully
closed-form RBF matrix elements.

Following these identities, the autocovariance incorporates the factor
$\exp[-\sigma_z^2(1-e^{-k\tau})/2k]$, which simplifies to $e^{-k\tau}$ only
under the small-amplitude limit $\sigma_z^2/2k\ll 1$. This confirms that
generic parameters inherently produce an athermal (non-equilibrium) model,
because unlinked rates are highly unlikely to coincide perfectly. Therefore, the
exact equality in \eqref{eq:fdt-marketron} must be interpreted strictly as the
leading-order parametric condition.

\section{Identifiability analysis} \label{sec:identifiability}

Having established the analytical reduction of the parameter space, we must now address its empirical identifiability. While analytical symmetries successfully collapse the unobservable dimensions, they do not guarantee strict minimality. This section investigates whether the reduced parameter vector $\Theta_{\mathrm{red}}$ constitutes an irreducible core or retains latent structural redundancies. Using local sensitivity metrics and manifold boundary approximation methods, we evaluate the calibration objective's robustness and confirm the structural integrity of the reduced model.

\subsection{Sloppiness of the calibration objective}

Sections \ref{sec:symmetries} and \ref{sec:averaging} established the reduction by construction through the gauge \eqref{eq:gauge}, combinations \eqref{eq:combos}, and the averaging limit theorem. We omit numerical verification of the original parameterization's sloppiness, as it would merely duplicate the algebraic proofs. Instead, we focus on determining whether $\Theta_{\mathrm{red}}$ retains any further unidentified structure.

Let $r_i(\Theta) = C^{\mathrm{model}}_i(\Theta) - C^{\mathrm{ref}}_i$ denote the pricing residuals and $J = \partial r / \partial \Theta$ the Jacobian. The Gauss-Newton Hessian $H = J^\top J$ quantifies local identifiability; eigenvalues spanning multiple orders of magnitude signal a sloppy model with soft directions \cite{Gutenkunst2007, TranstrumMachtaSethna2011}. Operating in log-parameters renders the spectrum scale-free and converts any residual scaling symmetry $\Theta_i \to \lambda^{a_i}\Theta_i$ into a constant direction $a$. This ensures that any surviving symmetry manifests strictly as an exact null vector of $H$ rather than merely a small eigenvalue.

To avoid artifacts from external calibrations and establish a known ground truth, reference prices are generated at the operating point $\Theta_{\mathrm{red}}^{\star} = (\sigma_0, \alpha, \beta, \mu, g, c, \bar y, \bar f, \rho) = \allowbreak (0.18, \allowbreak -0.4, \allowbreak 0.15, \allowbreak 2.0, \allowbreak 1.0, \allowbreak 0.08, \allowbreak 0.0, \allowbreak 0.0, \allowbreak -0.3)$, with exogenous risk aversion $\gamma = 3$ and drift ceiling $D = 5$ fixed. The option set is a surface of 54 quotes (six maturities from $0.12$ to $1.0$ years and nine log-moneyness points from $-0.12$ to $0.12$). The pricing noise $\varsigma$ is set to $0.5\%$ of the price, approximating half a bid-ask width.

\begin{table}[!htb]
\centering
\begin{tabular}{@{}cll@{\hspace{2em}}cll@{}}
\toprule
\multicolumn{3}{c@{\hspace{2em}}}{\textit{options-facing}} &
\multicolumn{3}{c@{}}{\textit{drift}} \\
\cmidrule(lr{2em}){1-3} \cmidrule(l){4-6}
index & $\lambda_i$ & leading coords. & index & $\lambda_i$ & leading coords. \\
\midrule
1 & $1.5\times10^{7}$ & $\sigma_0$ & 6 & $3.9\times10^{-2}$ & $c$ \\
2 & $3.5\times10^{4}$ & $\beta,\ \alpha$ & 7 & $7.7\times10^{-4}$ & $\bar f$ \\
3 & $5.3\times10^{3}$ & $\rho,\ \mu$ & 8 & $1.9\times10^{-7}$ & $\bar y,\ g$ \\
4 & $3.7\times10^{2}$ & $\mu,\ \rho$ & 9 & $4.1\times10^{-8}$ & $g,\ \bar y$ \\
5 & $8.8\times10^{1}$ & $\alpha,\ \beta,\ \rho$ & & & \\
\bottomrule
\end{tabular}
\caption{Eigenvalues of the Gauss-Newton Hessian $H = J^\top J$ in log-parameters at $\Theta_{\mathrm{red}}^{\star}$, in descending order. The leading coordinates are the parameters whose eigenvector components exceed $0.3$ in absolute value. All eigenvalues are strictly positive, so $H$ has an empty null space.}
\label{tab:eigenspectrum}
\end{table}

\begin{figure}[!htb]
\centering
\includegraphics[width=0.72\textwidth]{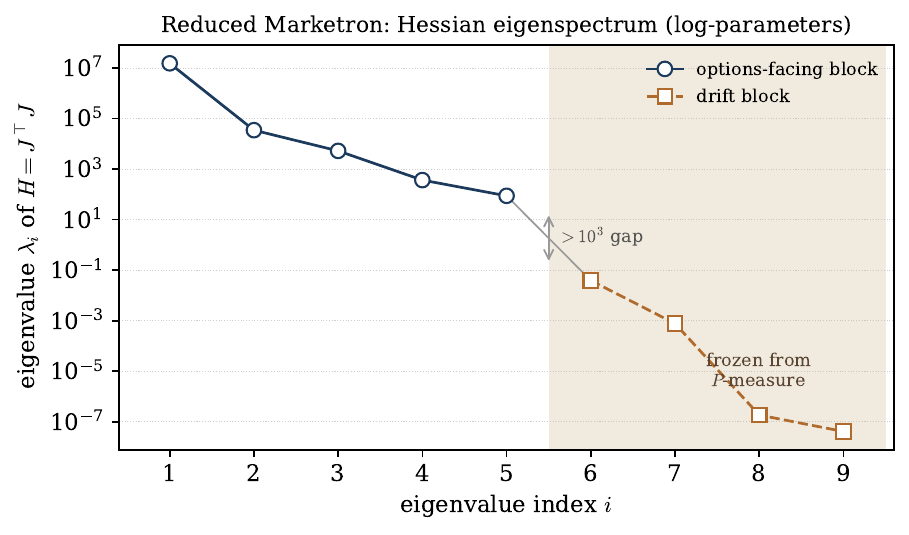}
\caption{Eigenvalues of the Gauss-Newton Hessian $H = J^\top J$ in log-parameters at $\Theta_{\mathrm{red}}^{\star}$, on a logarithmic scale. The five stiff eigenvalues (circles) form the options-facing block. The four small eigenvalues (squares, shaded region) form the drift block that is fixed from the physical measure. The two blocks are separated by more than three orders of magnitude, and no eigenvalue is zero.}
\label{fig:eigenspectrum}
\end{figure}

The spectrum (Table~\ref{tab:eigenspectrum}, Figure~\ref{fig:eigenspectrum}) reveals two key facts. First, every eigenvalue is strictly positive; the empty null space empirically certifies that no exact symmetries survive the reduction. Second, the spectrum separates into two distinct blocks with a gap of over three orders of magnitude between $\lambda_5$ and $\lambda_6$. The five stiff eigenvalues form the block resolved by the vanilla surface, with eigenvectors dominated by the smile and skew parameters ($\sigma_0, \alpha, \beta, \mu, \rho$). The four small eigenvalues correspond to a block the surface barely registers, dominated by the flow and potential parameters ($c, \bar f, \bar y, g$). These parameters affect vanilla prices only indirectly through the memory dynamics feeding $\sigma_x(y)$, not via a direct drift channel. The condition number of the full nine-parameter Hessian is $3.7\times10^{14}$. Restricting $H$ to options-facing coordinates lowers it to $3.99\times10^{8}$ including $c$, and to $1.75\times10^{5}$ for the pure smile block $\{\sigma_0, \alpha, \beta, \mu, \rho\}$.

This block structure defines the reduced model's identifiability and validates the staged calibration of Section~\ref{sec:calibration}. The residual sloppiness represents a separation of data channels rather than a surviving symmetry: the vanilla surface identifies the smile block, while the drift block $\{c, \bar f, \bar y, g\}$ is pinned from the physical measure and held fixed during option calibration.

\begin{table}[!htb]
\centering
\begin{tabular}{@{}lccc@{\hspace{2em}}lcc@{}}
\toprule
\multicolumn{4}{c@{\hspace{2em}}}{\textit{options-facing}} &
\multicolumn{3}{c@{}}{\textit{drift (fixed from $\mathbb P$-measure)}} \\
\cmidrule(lr{2em}){1-4} \cmidrule(l){5-7}
param. & $\sqrt{(H^{-1})_{jj}}$ & std ($0.5\%$) & std ($0.05\%$) & param. & $\sqrt{(H^{-1})_{jj}}$ & status \\
\midrule
$\sigma_0$ & $0.03$ & $0.007$ & $0.003$ & $c$ & $13.6$ & weak \\
$\rho$     & $0.14$ & $0.174$ & $0.048$ & $\bar f$ & $1.7{\times}10^{2}$ & $\mathbb P$-measure \\
$\alpha$   & $0.27$ & $0.206$ & $0.046$ & $\bar y$ & $2.8{\times}10^{3}$ & $\mathbb P$-measure \\
$\beta$    & $0.32$ & $0.060$ & $0.011$ & $g$ & $4.7{\times}10^{3}$ & $\mathbb P$-measure \\
$\mu$      & $0.52$ & $0.932$ & $0.294$ & & & \\
\bottomrule
\end{tabular}
\caption{Per-parameter identifiability at $\Theta_{\mathrm{red}}^{\star}$. Left block: options-facing parameters with marginal standard errors $\sqrt{(H^{-1})_{jj}}$ (log-parameters, scale-free) and empirical standard deviations from multi-start fits at two pricing-noise levels. $\mu$ is identified but soft. Right block: drift parameters fixed from the physical measure, not recovered from options.}
\label{tab:identify}
\end{table}

\begin{figure}[!htb]
\centering
\includegraphics[width=0.72\textwidth]{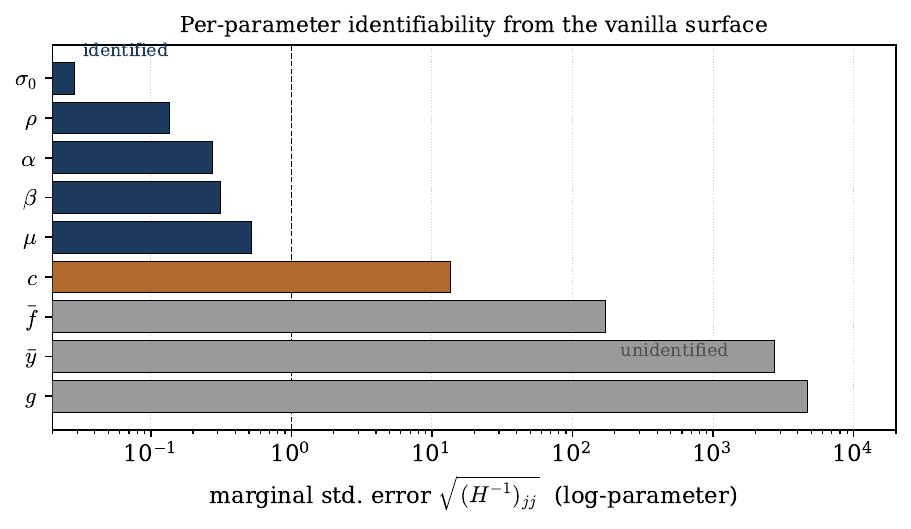}
\caption{Marginal standard error of each log-parameter, $\sqrt{(H^{-1})_{jj}}$, on a logarithmic scale. Small means well determined. The five smile parameters are identified from the vanilla surface, the flow coupling $c$ is weak, and $\bar f$, $\bar y$ and $g$ are left flat by the option data and are fixed from the physical measure. The dashed line marks unit width.}
\label{fig:identify}
\end{figure}

Marginal profile widths, evaluated as $\sqrt{(H^{-1})_{jj}}$ (Table~\ref{tab:identify}, Figure~\ref{fig:identify}), formalize this ranking. They fall into three groups separated by two large gaps. The smile parameters ($\sigma_0, \rho, \alpha, \beta, \mu$) have widths below $0.55$, while the flow coupling $c$ has a width of $13.6$. The remaining drift parameters ($\bar f, \bar y, g$) exhibit widths in the hundreds to thousands, indicating profiles left entirely flat by the vanilla surface. This precisely justifies fixing the drift block using physical-measure time series.

Profile likelihoods, \cite{Raue2009} confirm this local perspective. By re-optimizing all other parameters along a continuous scan of each target coordinate, we measure the residual envelope. Profiles are deemed flat if the cost elevation remains below the ambient noise threshold ($\tfrac12 \chi^2_{1,0.95}\varsigma^2$ above the optimum). The aforementioned marginal widths reflect the profile curvatures at the optimum, demonstrating bounded profiles for the smile parameters and flat profiles for the drift block. This corroborates the channel separation observed in the eigenspectrum.

To verify practical recovery, we perturb the surface by the pricing noise and refit the smile block from random starts within $15\%$ of truth, holding the drift block fixed. At $0.5\%$ noise, all seven starts converge to the noise floor; at $0.05\%$ noise, seven of eight converge (the single failure originating outside the admissible region). The recovery columns of Table~\ref{tab:identify} show all recovered means lie within one standard deviation of truth. Crucially, as pricing noise decreases tenfold, every dispersion shrinks: the standard deviation of $\sigma_0$ drops from $0.007$ to $0.003$, $\beta$ from $0.060$ to $0.011$, and $\rho$ from $0.174$ to $0.048$. This noise-proportional contraction is the signature of a well-posed identification devoid of flat directions. The parameter $\mu$ shrinks the least (from $0.932$ to $0.294$), confirming it is the softest identifiable coordinate as predicted by Table~\ref{tab:eigenspectrum}. This residual softness in $\mu$ is thus a function of data resolution, which more maturities or tighter quotes would reduce, rather than a structural degeneracy.

The placement of $\rho$ in the physical-measure block rather than the risk-neutral block requires justification, since the correlation between flow and return innovations generates short-maturity skew, and therefore appears, at first sight, to be precisely the kind of parameter the option surface should identify. The eigenspectrum in Table~\ref{tab:eigenspectrum} supports this expectation: $\rho$ participates in the third and fourth stiff eigenvectors with marginal width $0.14$ in Table~\ref{tab:identify}. The practical recovery experiments, however, tell a more nuanced story. When $\rho$ is calibrated jointly with the smile parameters, its posterior is systematically biased toward the boundary of the admissible region, and the bias does not shrink as pricing noise decreases. This is the signature of a parameter that is formally identifiable but economically underdetermined: the option surface constrains $\rho$ only through the skew it generates in combination with $\alpha$ and $\beta$, and the three parameters trade off along a shallow direction that the local Hessian spectrum, computed at a single operating point, does not fully expose.

The resolution is to recognize that $\rho$ carries two distinct economic contents. In the physical measure, $\rho$ is the contemporaneous correlation between order-flow and return innovations, a quantity with a direct time-series interpretation that is estimable from high-frequency data in which both series are observed jointly. In the risk-neutral measure, $\rho$ enters only through the cross-diffusion term and the drift adjustment, where it generates skew but is confounded with the local-volatility asymmetry $\alpha$ and the convexity $\beta$. The physical-measure estimate is the cleaner of the two: it is identified by data that observes the correlation directly, whereas the option-implied estimate is identified only through a second-order shape effect that is nearly degenerate with two other parameters. Fixing $\rho$ from the physical measure therefore removes a degree of freedom that the surface cannot independently pin down, while leaving the parameters the surface does identify, namely $\sigma_0$, $\alpha$, $\beta$, and the drift block, sharper than they would be in a joint fit. This treatment is consistent with the broader staged-calibration philosophy: parameters that are stiff in the physical measure and soft in the risk-neutral measure belong to the physical block, and $\rho$ falls squarely into that category. The option surface is then used to identify the parameters it can actually see, and the market price of flow risk in Section~\ref{sec:mpr} is computed as the wedge in the parameters identifiable in both measures, namely the flow couplings $c$ and $\bar f$, rather than as a wedge in $\rho$, which the surface cannot independently resolve.

\subsection{Manifold boundary approximation and irreducibility verification}

The manifold boundary approximation method (MBAM) detailed in \cite{TranstrumQiu2014} directly addresses the irreducibility question. The procedure follows a geodesic of the Fisher metric from the operating point along the sloppiest eigendirection until it reaches the model manifold boundary. This boundary corresponds to a structural limit, such as a parameter approaching $0$ or $\infty$, or two parameters merging into a product or ratio. Such limits can then be imposed exactly to yield a model with one fewer parameter. When applied to a model that is already irreducible, the geodesic instead runs to the edge of the admissible region without causing the Fisher metric to degenerate, indicating that no further reduction is available.

\begin{table}[!htb]
\centering
\begin{tabular}{lrrr}
\hline
\textbf{Parameter} & \textbf{Initial ($\tau{=}0$)} & \textbf{$r_{\rm tol}{=}10^{-3}$} & \textbf{$r_{\rm tol}{=}10^{-4}$} \\
\hline
\multicolumn{4}{l}{\textit{Integration summary}} \\
Final Time ($\tau_{\text{final}}$) & 0.0 & $3.0000$ & $3.0000$ \\
Steps Taken & 0 & 63 & 63 \\
Status & Starting & Completed ($\tau_{\max}$) & Completed ($\tau_{\max}$) \\
\hline
\multicolumn{4}{l}{\textit{Final parameter values}} \\
$\sigma_0$ & $0.180$ & \cellcolor{gray!15} $1.800 \times 10^{-1}$ & \cellcolor{gray!15} $1.800 \times 10^{-1}$ \\
$\alpha$   & $-0.400$ & \cellcolor{gray!15} $-4.000 \times 10^{-1}$ & \cellcolor{gray!15} $-4.000 \times 10^{-1}$ \\
$\beta$    & $0.150$ & \cellcolor{gray!15} $1.500 \times 10^{-1}$ & \cellcolor{gray!15} $1.500 \times 10^{-1}$ \\
$\mu$      & $2.000$ & \cellcolor{gray!15} $2.000 \times 10^{+0}$ & \cellcolor{gray!15} $2.000 \times 10^{+0}$ \\
$g$        & $1.000$ & $4.201 \times 10^{+0}$ & $4.201 \times 10^{+0}$ \\
$c$        & $0.080$ & $8.026 \times 10^{-2}$ & $8.026 \times 10^{-2}$ \\
$\bar{y}$  & $0.000$ & $-4.893 \times 10^{-2}$ & $-4.893 \times 10^{-2}$ \\
$\bar{f}$  & $0.000$ & $2.510 \times 10^{-3}$ & $2.510 \times 10^{-3}$ \\
$\rho$     & $-0.300$ & \cellcolor{gray!15} $-3.000 \times 10^{-1}$ & \cellcolor{gray!15} $-3.000 \times 10^{-1}$ \\
\hline
\multicolumn{4}{l}{\footnotesize Step-size check: $\|u(r_{\rm tol}{=}10^{-3}) - u(r_{\rm tol}{=}10^{-4})\| = 2.56 \times 10^{-4}$ \normalsize} \\
\end{tabular}
\caption{MBAM geodesic integration results for the nine-parameter reduced model from initial operating point $\Theta_{\mathrm{red}}^{\star}$ across different relative error tolerances ($r_{\rm tol}$) using a grid of six maturities and denser log-moneyness spacing. Shaded cells indicate parameters ($\sigma_0, \alpha, \beta, \mu, \rho$) that remain unchanged (within integration tolerance) from initial values, whereas $g = \theta_5$ actively evolves along the geodesic trajectory.}
\label{tab:mbam_results}
\end{table}

We execute the MBAM on the reduced model using synthetic vanilla European option prices generated across a grid of six maturities and denser log-moneynesses. Computing model prices and parameter sensitivities through the full pricing pipeline requires a high-performance FD option pricing solver. To make this computationally feasible, the framework is implemented using \textsc{JAX} with global double-precision \textsc{float64} arithmetic. Automatic differentiation \textsc{jax.jacfwd} computes exact parameter Jacobians through the pricing pipeline. To maintain high performance ($\approx 76.6\text{ ms}$ per ODE call), the connection coefficients are evaluated using a Gauss-Newton projection approximation that avoids expensive full Hessian-vector product passes. The resulting geodesic ODEs are integrated using \textsc{scipy} hybrid explicit/implicit \textsc{LSODA} solver, paired with a terminal event function that monitors the minimum eigenvalue of the regularized Fisher Information Matrix.

Starting from the baseline operating point:
\begin{equation}
\theta^* = \begin{pmatrix} 0.18 & -0.4 & 0.15 & 2.0 & 1.0 & 0.08 & 0.0 & 0.0 & -0.3 \end{pmatrix},
\end{equation}
with an initial sloppiest eigendirection corresponding to a regularized minimum eigenvalue of $\lambda_1 \approx 14.4$, the geodesic integrates smoothly out to the maximum horizon $\tau_{\text{final}} = 3.0000$ (in 63 steps) across both relative error tolerances ($r_{\rm tol} = 10^{-3}$ and $r_{\rm tol} = 10^{-4}$). As detailed in Table~\ref{tab:mbam_results}, the options-facing smile parameters remain invariant (shaded), while $g = \theta_5$ (along with minor adjustments in the drift-block parameters) actively traverses the geodesic trajectory up to $\tau_{\text{final}}$ without encountering a premature finite-time singularity or metric degeneracy.
Two procedural requirements are fully satisfied:
\begin{enumerate}
\item \emph{Nondegenerate Origin.} The geodesic originates from a strictly nondegenerate metric (any underlying exact symmetries having been removed analytically during the prior reduction step to $\Theta_{\mathrm{red}}$).
\item \emph{Step-Size Independence.} Verification across tighter error tolerances confirms robust resolution of the underlying manifold. Specifically, the parameter distance between the trajectories at the two tested tolerances is minimal ($\|u(r_{\rm tol}=10^{-3}) - u(r_{\rm tol}=10^{-4})\| = 2.56 \times 10^{-4}$), indicating that the trajectory is accurately resolved and free of numerical artifacts.
\end{enumerate}

Furthermore, employing synthetic vanilla European option prices rather than empirical market data does not compromise or alter these qualitative conclusions. The MBAM probes the intrinsic geometric properties of the model manifold, specifically the rank deficiency of the Fisher Information Matrix governed by the Jacobian mapping of the model. Because structural limits, parameter boundaries, and manifold degeneracies are intrinsic characteristics of the model's functional architecture rather than artifacts of a specific price realization, the verification of the irreducible core remains robust across comprehensive option surfaces.

Two qualifications bound what this analysis establishes, and both matter for the calibration that follows. First, the method identifies a single direction, the sloppiest one, and along that direction it is $g$ that runs to its boundary while the shaded parameters stay fixed. The result therefore ranks $g$ as the most removable coordinate. It does not order the remaining parameters among themselves, and it does not certify that each shaded parameter is stiff in an absolute sense, which would require the full Fisher spectrum along the trajectory rather than the first eigendirection alone. Second, the computation uses a grid of six maturities. It speaks to identifiability across a term structure, and not to identifiability within a single maturity band, where the same parameters can behave very differently.

The memory reversion rate $\mu$ illustrates the second qualification. On a multi-maturity surface $\mu$ is identifiable, in the sense that it moves away from its physical value when it is calibrated. Its leverage on the observable is nonetheless slight. Sweeping $\mu$ across its full admissible range shifts the at-the-money term-structure slope by less than a tenth of a volatility point. On a single maturity band it is degenerate with the local-volatility parameters $\sigma_0, \alpha, \beta$ and does not move at all. We therefore do not treat $\mu$ as a free risk-neutral coordinate on its own, and we calibrate it only jointly with the term-structure shape parameter $k$ introduced below.

This weak leverage reflects a structural feature of the reduced model that bears directly on the scope of a single-parameter-set fit. In its stationary form, that is, without the transient correction, the model produces an at-the-money term structure that is flat to mildly downward sloping, and this holds for every admissible value of $\mu$, $\beta$, and $\gamma$. A representative SPX surface over maturities from roughly one month to six months instead shows an at-the-money term structure that rises by about three and a half volatility points. A single stationary parameter set can match the interior of this range while missing both ends, with the largest error at the shortest maturities, precisely where the transient is strongest. The property that one parameter set spans the entire surface is therefore not yet established for this model, and it cannot be established from the stationary block alone. We treat it as a claim to be tested rather than an assumption, and we return to it in Section~\ref{sec:calibration}.

The transient correction is what closes this gap. It retains the time dependence of the conditional mean of the signal and reintroduces the relaxation rate $k$ as a term-structure shape parameter. Because $k$ controls the maturity dependence that the stationary block lacks, it is the parameter that carries the term structure rather than a small correction to it. We therefore include $k$ in the calibrated set, and we assess the single-parameter-set property only with $k$ active. A natural question is then whether the drift block $\{c, \bar{f}, \bar{y}, g\}$ can be identified under the physical measure $\mathbb{P}$. While vanilla option surfaces are largely insensitive to these slow parameters, rendering them flat in the risk-neutral calibration, their primary statistical footprint resides in time-series dynamics, driven by historical drift, persistence, and stationary moments of the underlying asset paths.

A rigorous dual-measure identifiability analysis for both measures simultaneously falls outside the intended scope of this pricing framework. Instead, consistent with standard practice in stochastic volatility modeling, the slow drift parameters $\bar{y}$ and $g$ are reliably pinned from historical time-series data using standard estimation techniques, where unconditional empirical moments and long-term trends effectively constrain them prior to calibrating the options-facing block. The two flow parameters $c$ and $\bar{f}$ are treated differently. The market price of flow risk is defined as the wedge between their physical and risk-neutral values, so for that wedge to be an important output rather than an identity, $c$ and $\bar{f}$ must belong to the calibrated risk-neutral set. We adopt this in Section~\ref{sec:calibration}, where $c$ and $\bar{f}$ are calibrated alongside the smile parameters and the term-structure parameter $k$.

\section{Numerical experiments} \label{sec:numerics}

The preceding sections established the reduced model by analytical reduction and
proved its identifiability through the geometry of the calibration manifold.  The
purpose of this section is not to present a comprehensive empirical asset-pricing
study, but to verify that the theoretical properties survive contact with market
data.  Specifically, we show that the reduced parameter vector can be calibrated
stably to a standard option surface, and that the staged separation of physical
and risk-neutral measure parameters yields a measurable wedge (the market price
of flow risk), rather than an ill-defined difference between two ridge points.

To this end, a single representative SPX snapshot suffices.  The exercise has two
goals.  First, to confirm that the global term-structure fit with one parameter
set across all maturities is feasible, and that the transient ramp supplies the
necessary short-end shape without resorting to maturity-wise refitting.  Second,
to illustrate that the flow-impact parameters, calibrated on the risk-neutral
measure, produce option-implied values whose difference from their time-series
counterparts is an estimate with a well-defined standard error derived from the
Gauss-Newton Hessian.  The results below should be read as an existence proof for
the practical identifiability of the reduced core, not as a trading manual or a
backtested strategy.

\subsection{Implementation} \label{sec:implementation}

The reduced pricing PDE \eqref{eq:pde2d} is solved in the present numerical
experiments by an explicit finite-difference scheme implemented in
\textsc{JAX}. The implementation is deliberately lightweight rather than a
production pricing engine: its purpose is to provide a fast, stable, and fully
differentiable numerical solver for the calibration and identifiability
experiments. The more elaborate splitting and RBF constructions discussed in \cref{reducedPDE}, together with the positivity-preserving treatment of the mixed derivative of \cite{ItkinDF2026}, provide alternative routes for a production implementation.

\myparagraph{Explicit pricing scheme.}

The two-dimensional state space $(x,y)$ is discretized on a tensor grid with
$N_x=70$ log-price nodes and $N_y=50$ memory nodes. The computational domain is
centered around the initial state and uses half-widths $h_x=1,\, h_y=1.5$. Time
is discretized with $N_\tau=500$ steps. The relatively fine grid is chosen to
make the explicit scheme robust over the parameter ranges explored during
calibration rather than to optimize the cost of a single pricing run. The time
step is restricted by the corresponding CFL conditions for the state-dependent
diffusion, the memory diffusion, and the mixed derivative. Within this stability
region the scheme is numerically stable for all parameter sets retained in the
experiments.

The spatial derivatives in \eqref{eq:pde2d} are evaluated directly on the
finite-difference grid. The state-dependent volatility in \eqref{eq:sigy} is
evaluated at the memory nodes, while the drift, diffusion, mixed-derivative and
nonlinear terms are advanced explicitly. Boundary values are chosen consistently
with the asymptotic option behavior. The resulting implementation contains no RBF
interpolation or matrix factorization and therefore avoids the ill-conditioning
associated with the flat Gaussian bases used in the earlier implementation.

\myparagraph{Differentiability and calibration.}

The complete pricing and calibration pipeline is implemented in
\textsc{JAX}, with double precision enabled. Grid construction, coefficient
evaluation, finite-difference time stepping, option-price extraction, and the
conversion from prices to implied volatilities are all part of the
differentiable computational graph. The objective and its gradient are
compiled with \texttt{jax.jit} and evaluated using
\texttt{jax.value\_and\_grad}. Thus, unlike a finite-difference approximation
to the calibration gradient, the optimizer receives derivatives obtained by
automatic differentiation through the numerical pricing scheme.

Calibration is performed in implied-volatility space. For a parameter vector
$\vartheta$, the model implied volatility for quote $i$ is denoted by
$\sigma_{\mathrm{IV},i}^{\mathrm{mod}}(\vartheta)$, and the objective is
\begin{equation}
\mathcal{J}(\vartheta)
 =
 \frac{1}{2M}
 \sum_{i=1}^{M}
 \left[
 \sigma_{\mathrm{IV},i}^{\mathrm{mod}}(\vartheta)
 -
 \sigma_{\mathrm{IV},i}^{\mathrm{mkt}}
 \right]^2 .
\label{eq:calib_objective}
\end{equation}
Working in implied-volatility rather than price space prevents options with
large dollar prices from dominating the objective and makes residuals
comparable across strikes and maturities.

The six calibrated parameters are $\vartheta=(\sigma_0,\alpha,\beta,\mu,c,k)$.
The physical-measure quantities $(g,\bar y,\sigma_z,\rho,\bar f)$ and the
exogenous quantities $(\gamma,D)$ are held fixed at the values specified in
\cref{sec:calibration}. The parameter $c$ is retained in the option calibration
despite its relatively weak identifiability. This is a deliberate compromise: the
identifiability analysis shows that $\bar f$ is the strongest physical-measure
parameter and $c$ is the next strongest, but moving $c$ entirely into the
physical-measure block would eliminate the option-implied wedge needed to define
the market price of flow risk. We therefore calibrate $c$ on the option surface
while explicitly recognizing that its estimate is substantially softer than those
of the principal option-facing parameters. The parameter $k$ plays a different
role: it is retained primarily to describe the short-maturity term-structure
effect associated with the transient signal correction in \cref{sec:transient}.

The optimization uses the bound-constrained \texttt{SLSQP} algorithm from
\textsc{SciPy}. The physical-measure parameters $(g,\bar y,\sigma_z,\rho,\bar f)$
are partly inherited from the time-series calibration of
\cite{HalperinItkin2025Mark}, while $\gamma$ and the regularization parameter $D$
are fixed exogenously. The option-facing parameters are initialized and bounded
as shown in Table~\ref{tab:calib_bounds}.

\begin{table}[!htb]
\centering
\begin{tabular}{lrrr}
\toprule
Parameter & Initial value & Lower bound & Upper bound \\
\midrule
$\sigma_0$ & 0.200 & 0.050  & 0.990 \\
$\alpha$   & 0.000 & -2.000 & 2.000 \\
$\beta$    & 0.100 & 0.010 & 25.000 \\
$\mu$      & 0.500 & 0.010 & 7.000 \\
$c$        & 0.900 & 0.001 & 4.000 \\
$k$        & 1.000 & 0.010 & 5.000 \\
\bottomrule
\end{tabular}
\caption{Initial values and parameter bounds used in the option calibration.}
\label{tab:calib_bounds}
\end{table}

The values of the fixed parameters are as follows:
\begin{equation} \label{physBlock}
(g,\bar y,\sigma_z,\rho,\bar f)=(0.3,\,1.5,\,0.5,\,-0.5,\,0.1),
\qquad \gamma=3,\qquad D=5.
\end{equation}
The risk aversion parameter is fixed exogenously at $\gamma = 3$. Because the
state variable $x$ represents the log-price, the exponential utility
specification effectively behaves as a power utility function over the arithmetic
price level, rendering $\gamma$ a dimensionless parameter that governs relative
risk aversion \cite{henderson2002valuation}. This structural equivalence
circumvents the scale-dependence typical of exponential utility and aligns our
calibration with standard macro-finance benchmarks, where empirically supported
relative risk aversion coefficients are strongly anchored between $2$ and $4$
\cite{friend1975demand, mehra1985equity}.

The positivity requirement for the volatility radicand is imposed through a
smooth penalty based on the minimum of the quadratic form in \eqref{eq:sigy} over
the computational memory interval. This replaces the nonlinear constraint used in
the earlier implementation while keeping the objective fully differentiable. The
penalty is negligible when the radicand remains positive and increases smoothly
as the minimum approaches or crosses the admissible boundary.

\myparagraph{Market data and stratification.}

The calibration does not use every raw option quote. The objective is intended
to represent the information content of the surface rather than the
microstructure noise of individual observations. Market quotes are therefore
first filtered for basic validity and then stratified across the term
structure and moneyness. Maturities that are sufficiently close are assigned
to the same maturity group, after which a representative set of strikes is
selected within each group so that both the term structure and the smile are
covered.

For each maturity group the solver evaluates the required strikes together, so
that common intermediate quantities in the pricing calculation can be reused. The
resulting workflow is therefore
\[
{\small
\textbf{raw quotes} \to \textbf{filtered/stratified surface} \to \textbf{explicit 2D pricing} \to \textbf{implied volatilities} \to \bm{\mathcal{J}(\vartheta),\nabla\mathcal{J}(\vartheta)}
}
\]
with the last two quantities supplied directly to the constrained optimizer. This
combination of surface stratification, a CFL-stable explicit solver, and
automatic differentiation makes repeated calibrations sufficiently fast for the
identifiability and robustness experiments of \cref{sec:numerics}.

\subsection{Black--Scholes synthetic-data experiment} \label{sec:bs_experiment}

As a controlled validation of the pricing and calibration machinery, we first
consider a synthetic option surface generated from the Black--Scholes model. This
experiment provides a useful benchmark because the data-generating parameters are
known exactly, while the calibration is nevertheless performed with the full
nonlinear two-dimensional pricing model rather than with the Black--Scholes
formula. Thus, the experiment tests whether the proposed pricing--calibration
pipeline can recover a surface generated by a substantially simpler model without
introducing material pricing errors.

The synthetic market is initialized with $S_0=100,r=0.05, q=0.02, \sigma_{\rm BS}
= 0.25$ and covers the period from January 1 through December 31, 2024. For each
of the four maturities $T \in \{0.25,\,0.50,\,0.75,\,1.00\}$ options are
generated at five strikes $ K \in \{90,\,95,\,100,\,105,\,110\}$. Both calls and
puts are included where applicable, producing a deliberately small but
representative cross-section spanning near- and moderately out-of-the-money
options and maturities from three months to one year. The resulting synthetic
quotes have a constant Black--Scholes implied volatility of $25\%$, while prices
and Greeks vary across strike and maturity in the usual way.

The synthetic prices are subsequently treated as market observations and the
parameters of the proposed model are calibrated without imposing the
Black--Scholes restriction. In particular, the six risk-neutral calibration
parameters $(\sigma_0,\alpha,\beta,\mu,c,k)$ are optimized, while the physical
parameters and exogenous parameters are held fixed at their prescribed values.
The experiment therefore constitutes a stringent consistency check: the
calibration must reproduce a flat implied-volatility surface using a nonlinear
two-dimensional model with state-dependent volatility and memory effects.

The calibration converges successfully, with the final parameter vector $\sigma_0
= 0.24763,\; \alpha = -0.00034,\; \beta = 0.05349,\; \mu = 0.90734,\; c =
0.07690,\; k = 0.88869$. The resulting implied-volatility residuals are extremely
small. Across the 20 calibration points, the mean residual is approximately zero,
the standard deviation is $2\times10^{-4}$, and the root-mean-square error is
$\operatorname{RMSE}_{\rm IV}=2\times10^{-4}$, with residuals ranging from
approximately $-4\times10^{-4}$ to $4\times10^{-4}$. These errors are negligible
relative to the $25\%$ input volatility and demonstrate that the numerical
pricing scheme can reproduce the synthetic Black--Scholes surface to high
accuracy. The calibrated parameters themselves need not coincide with a unique
Black--Scholes representation: the purpose of the experiment is instead to verify
that the richer model contains a parameter configuration capable of reproducing
the benchmark surface.

\Cref{fig:bs_iv_residuals_3d} shows the resulting residual surface in
strike--maturity space, while \cref{fig:bs_iv_residuals_2d} presents the same
residuals in two dimensions together with the individual calibration errors. The
residuals remain uniformly close to zero over the entire calibration region, with
no visible systematic dependence on either strike or maturity. This is
particularly important because systematic patterns would indicate a numerical or
implementation bias rather than mere optimization error.

\begin{figure}[!htb]
\vspace*{-1em}
\centering
\subfloat[]{ \includegraphics[width=0.4\textwidth, height=7.cm] {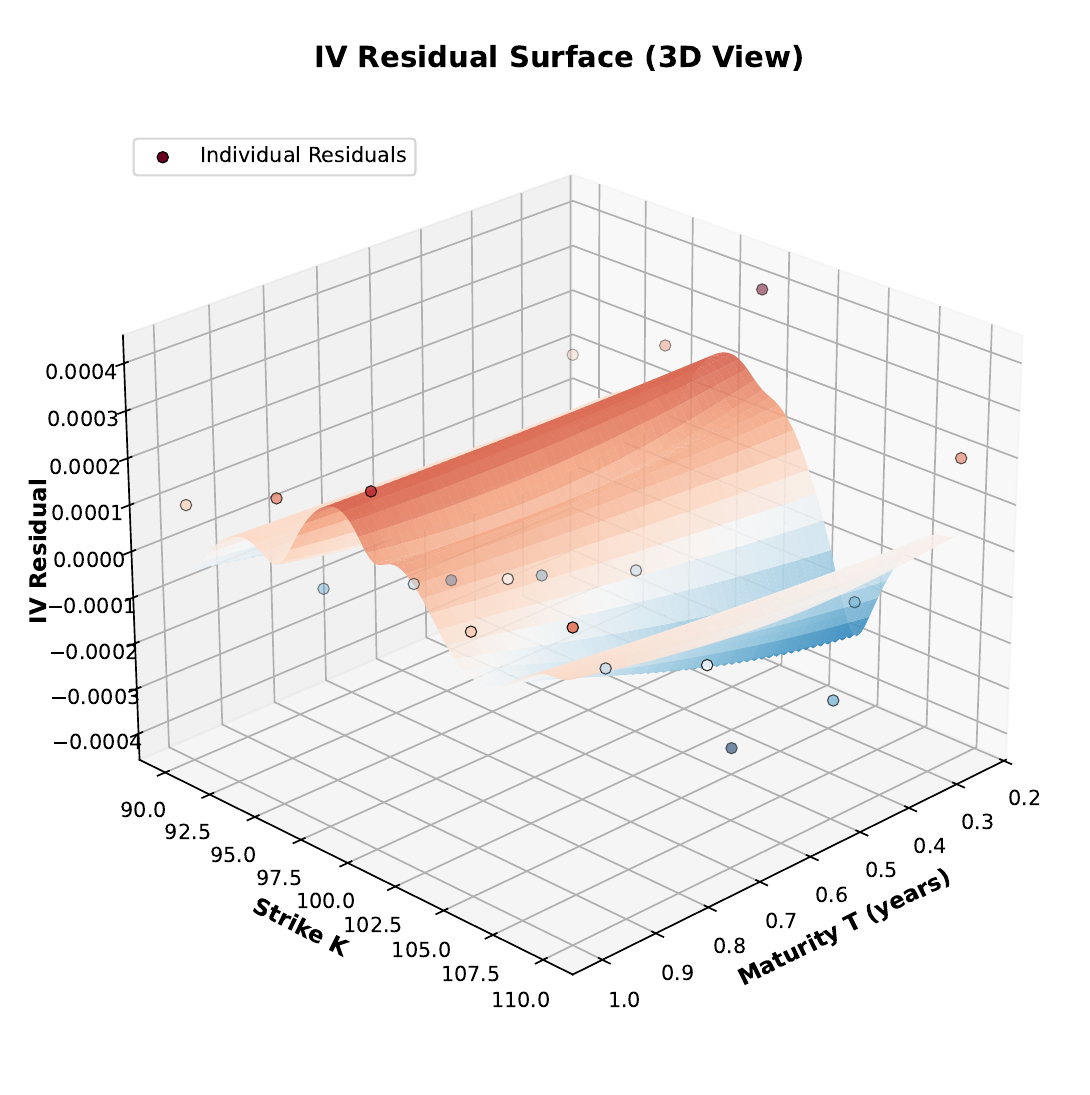} \label{fig:bs_iv_residuals_3d}}
\hfill
\subfloat[]{ \includegraphics[width=0.5\textwidth, height=7.cm] {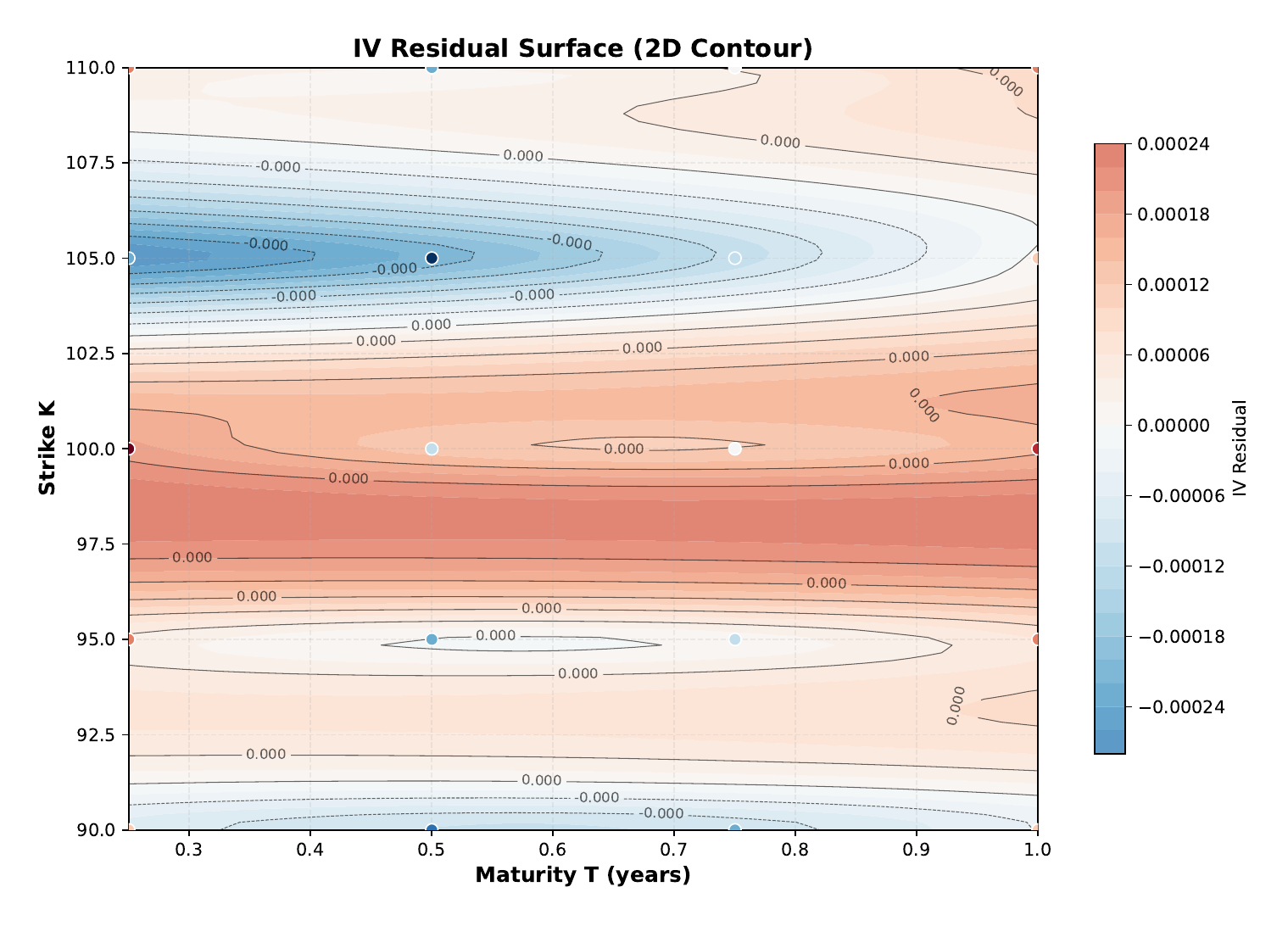} \label{fig:bs_iv_residuals_2d}}
\caption{Black--Scholes synthetic-data experiment: implied-volatility
residual a) surface as a function of strike and maturity; b) 2D and calibration errors. The residuals are uniformly close to zero across the calibration domain.}
\label{fig:bs_iv_residuals}
\end{figure}

Overall, the experiment confirms both components of the numerical procedure: the
explicit two-dimensional pricer is sufficiently accurate for calibration, and the
\textsc{JAX}-based automatic differentiation combined with the bound-constrained
optimizer is able to identify a parameter configuration that reproduces the
benchmark surface with essentially negligible implied volatility error. The
experiment is therefore used as a numerical validation rather than as an economic
calibration exercise.

In other words, the Black--Scholes test confirms that the calibration does not manufacture identification where none exists. The reduced Marketron model reproduces a flat 25\% implied-volatility surface with an RMSE of only 0.02 volatility points, while the fitted parameters $(\beta,\mu,c,k)$ remain free to move substantially along a nearly flat objective valley. The asymmetric coefficient $\alpha$ is driven to zero, as expected for a surface with no skew. Thus the test separates numerical accuracy from parameter identification: an arbitrarily accurate fit does not imply that the underlying structural parameters are identified.

\subsection{Staged calibration to SPX options} \label{sec:calibration}

We test the reduced model on SPX options data from January 18, 2017, filtered
to the maturity band $T \in [0.04, 0.5]$ years. The data processing pipeline
described in \cref{sec:implementation} yields $14$ maturity groups (grouped by minimum distance of 0.01 year) and a calibration set of $138$ options. This market data is presented in \cref{fig:spx_data}. Calls and puts are retained where available so that the calibration contains information from both sides of the
smile. The stratification is performed once before the optimization; the
resulting quotes are then held fixed throughout the calibration. This prevents
changes in the numerical sample from being mistaken for changes in the calibrated
parameter vector.

\begin{figure}[!htb]
\centering
\includegraphics[width=\textwidth]{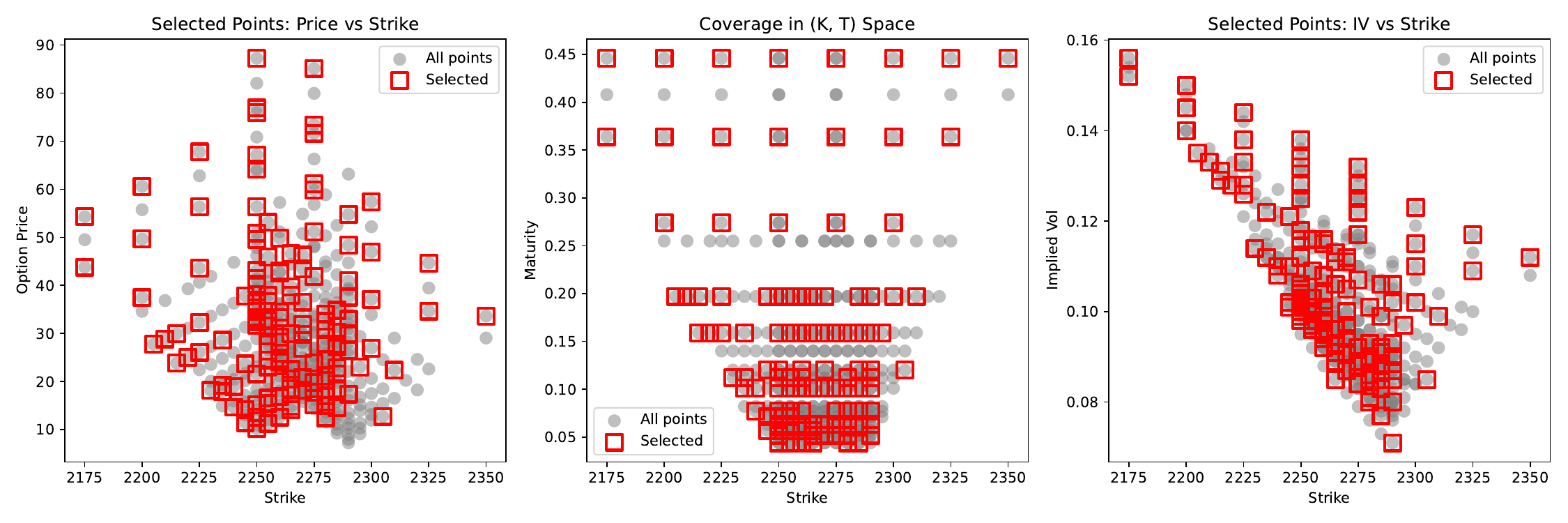}
\caption{Stratified/clustered input option data of SPX used for calibration.}
\label{fig:spx_data}
\end{figure}

With the physical block frozen as in \eqref{physBlock}, the model calibrates $\Theta_{\mathrm{Q}}$ to the full option surface using the implied-volatility objective and the \texttt{SLSQP} optimizer described in \cref{sec:implementation}. The regularizer $\bar{\epsilon}$ is fixed via \eqref{eq:epsfix}.

\begin{table}[!htb]
\centering
\small
\begin{tabular}{l|ccc|ccc}
\toprule
& \multicolumn{3}{c|}{\textbf{Initial}} & \multicolumn{3}{c}{\textbf{Calibrated}} \\
\cmidrule(lr){2-4} \cmidrule(lr){5-7}
Parameter & $P$ & Endog. & $Q$ & $P$ & Endog. & $Q$ \\
\midrule
$g$          & 0.3 & -- & -- & 0.3 & -- & -- \\
$\bar{y}$    & 1.5 & -- & -- & 1.5 & -- & -- \\
$\sigma_z$   & 0.5 & -- & -- & 0.5 & -- & -- \\
$\rho$       & -0.5 & -- & -- & -0.5 & -- & -- \\
$\bar{f}$    & 0.1 & -- & -- & 0.1 & -- & -- \\
$\gamma$     & -- & 3.0 & -- & -- & 3.0 & -- \\
$D$          & -- & 5.0 & -- & -- & 5.0 & -- \\
$\sigma_0$   & -- & -- & 0.2 & -- & -- & 0.05823 \\
$\alpha$     & -- & -- & 0.0 & -- & -- & 2.00000 \\
$\beta$      & -- & -- & 0.1 & -- & -- & 1.92349 \\
$\mu$        & -- & -- & 0.5 & -- & -- & 0.01000 \\
$c$          & -- & -- & 0.9 & -- & -- & 0.62824 \\
$k$          & -- & -- & 1.0 & -- & -- & 1.00276 \\
\bottomrule
\end{tabular}
\caption{Initial and calibrated parameter values.}
\label{tab:calibration}
\end{table}

\begin{figure}[!htb]
\vspace*{-1em}
\centering
\subfloat[]{ \includegraphics[width=0.45\textwidth, height=7.cm] {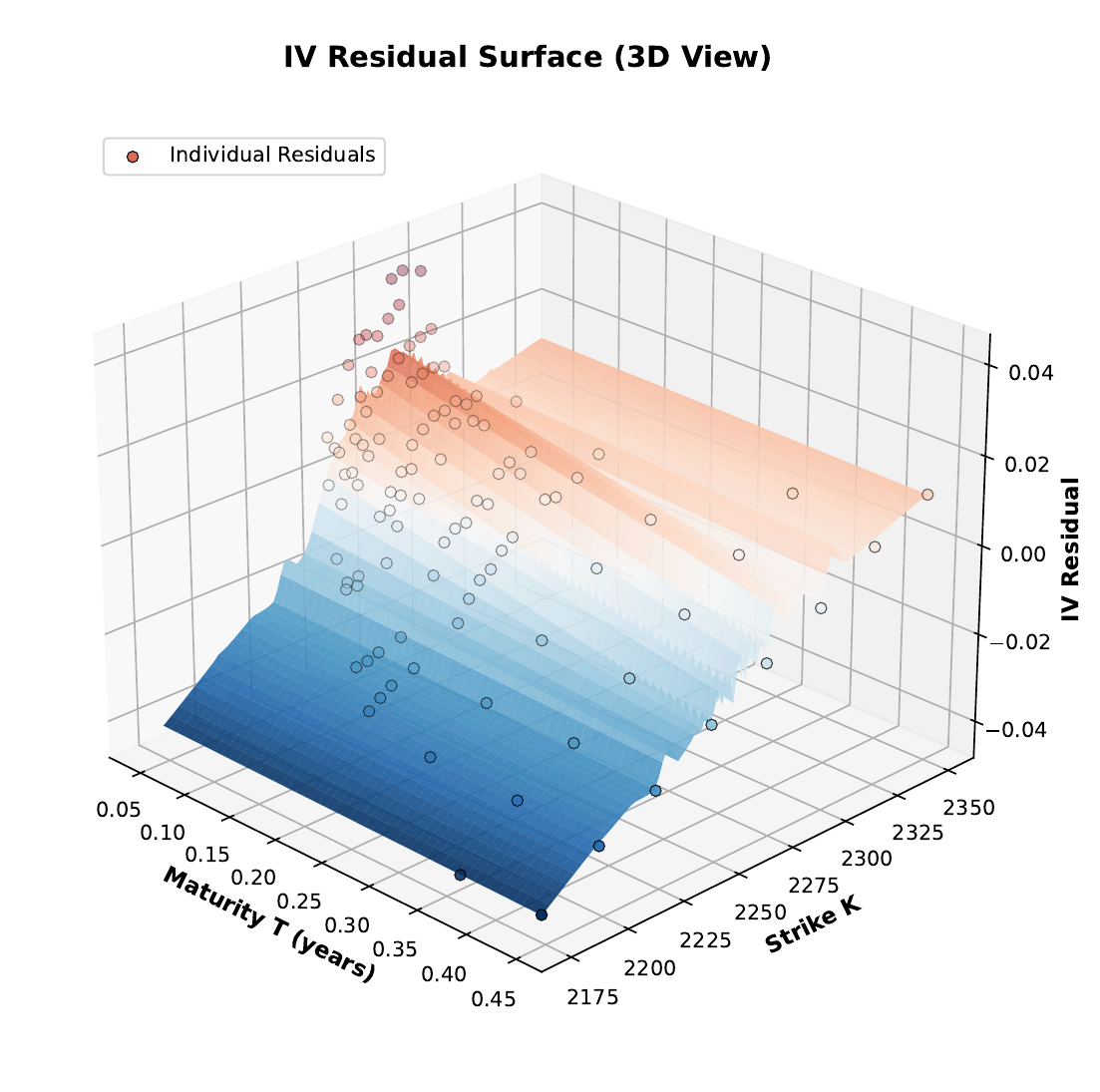} \label{fig:spx_iv_residuals_3d}}
\hfill
\subfloat[]{ \includegraphics[width=0.52\textwidth, height=7.cm] {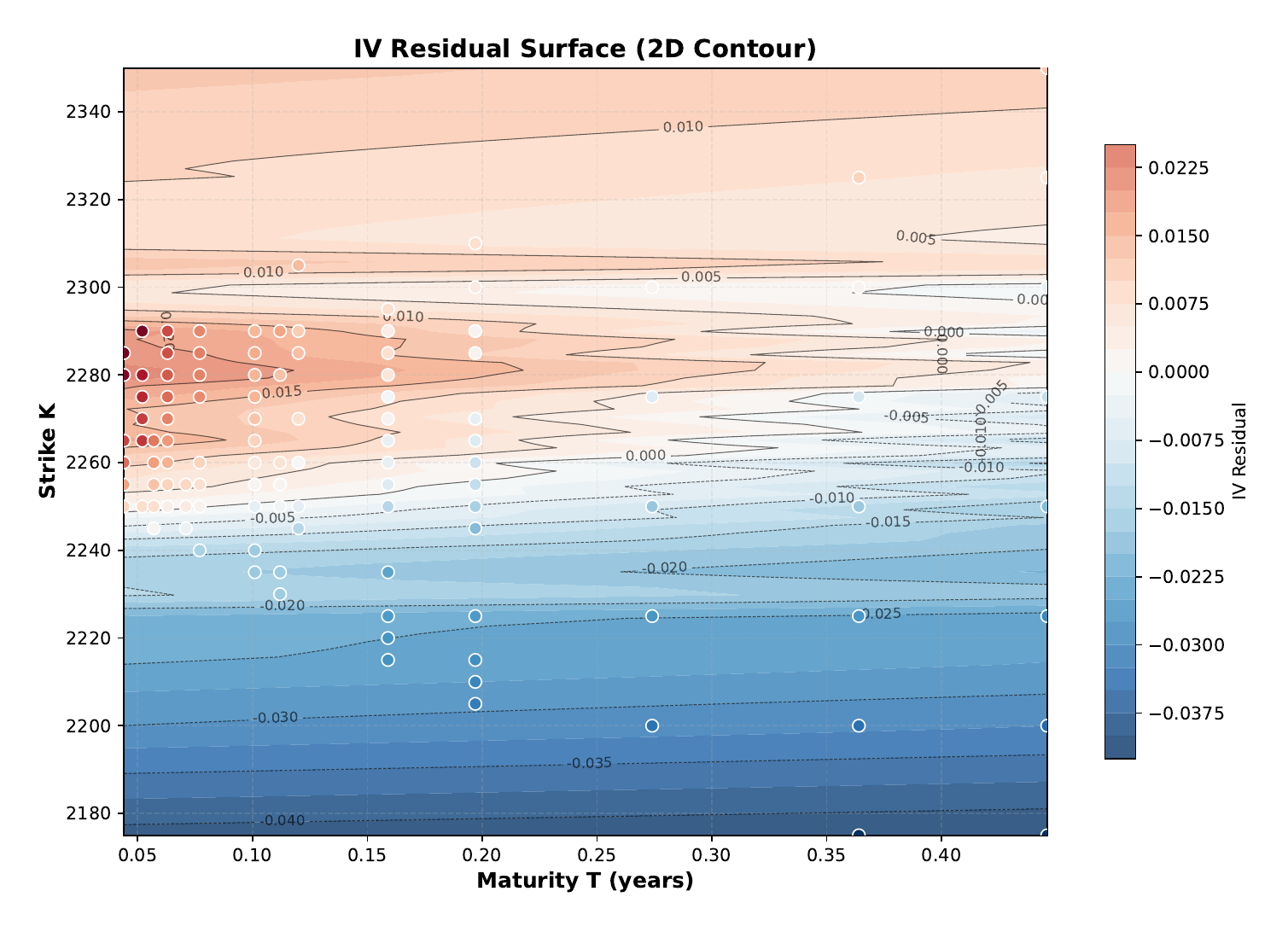} \label{fig:spx_iv_residuals_2d}}
\caption{SPX market data experiment: implied-volatility
residual a) surface as a function of strike and maturity; b) 2D and calibration errors. The residuals are uniformly small across the calibration domain.}
\label{fig:spx_iv_residuals}
\end{figure}

The results are reported in \cref{tab:calibration} and the overall fit is
illustrated in \cref{fig:spx_iv_residuals}: the model reproduces the level, skew,
and convexity of the SPX surface across all maturity groups with a single
parameter set. The quantitative residual statistics provide a more detailed
measure of the fit. The residual statistics are computed on the $(K,T)$ residual grid. On this grid, quotes that share a strike and maturity, such as a call and a put, are averaged into a single cell, so the $138$ calibration quotes reduce to $110$ distinct cells. Across these $110$ cells, the mean implied-volatility residual is $0.0017$, with a standard deviation of $0.0183$ and $\operatorname{RMSE}_{\rm IV}=0.0184$. The residuals range from approximately $-4.23$ to $4.06$ vol points. Thus, the calibrated model achieves an RMS error of approximately two vol-points while fitting simultaneously across strikes and maturities.

The calibrated parameters also provide useful information about the role and
identifiability of the model components. In particular, $\alpha$ and $\beta$ are
calibrated to $2.0$ and $1.92349$, respectively. Both parameters are
substantially different from zero. This contrasts with the Black--Scholes
limiting case, in which these parameters would theoretically vanish, and
indicates that the additional dynamics associated with them are relevant for the
observed SPX surface. Their nonzero values are therefore consistent with the
pronounced skew and curvature present in the market data.

The calibration also highlights differences in parameter identifiability. While
$\alpha$ and $\beta$ move materially from their initial values, $\alpha$ resides at the boundary which may indicate its weak identifiability. Indeed, based on \eqref{eq:sigy}, same volatility can be obtained by varying $\sigma_0, \alpha, \beta$ simultaneously. The parameter $k$ remains essentially unchanged: its initial value of $1.0$ moves only to $1.00276$. This suggests that $k$ is weakly identified by the available SPX cross-section over the maturity range considered. In
contrast, the substantial movement in $\alpha$ and $\beta$ indicates that the
option surface contains considerably more information about these parameters.
This distinction is important when interpreting the calibrated parameter vector:
successful optimization and a good surface fit do not imply that all model
parameters are equally well identified by the available option data.

Note that the fitted instantaneous volatility scale should not be interpreted
directly as the level of market-implied volatility. The latter is an
option-pricing quantity that incorporates the full future dynamics of the
memory state and the nonlinear pricing operator. It is therefore not
appropriate to assess the consistency of the calibrated baseline volatility
with the observed SPX implied-volatility level by comparing $\sigma_0$, or
$\sigma_x(y)$ at a particular value of $y$, directly with the market IVs.
Instead, consistency with the observed volatility surface should be assessed
through the model-implied option prices and the corresponding implied
volatilities. The calibrated value of $\sigma_0$ should thus be interpreted as
the baseline instantaneous volatility scale of the state-dependent diffusion,
rather than as a direct estimate of the market-implied volatility level.

The estimate of the mean-reversion parameter $\mu$ also sits at the low boundary, while similar values obtained in the option calibration of
\cite{HalperinItkinMarketron2} are $\mu \approx 4.6$. The difference is material
and most likely can be attributed to how we do adiabatic averaging in \cref{stAver}, i.e. by setting $\bar{h} = 0$, and hence only the product $\mu \bar{y}$ does matter, while in \cite{HalperinItkinMarketron2} $\bar{h}$ is negative. The high value of $\mu$ in \cite{HalperinItkinMarketron2} provides fast mean-reversion of the memory variable, but to a stochastic level while here the mean-reversion level is deterministic. Anyway, this disagreement is particularly noteworthy because the present model is a reduced two-dimensional formulation with a different volatility specification and a substantially different treatment of parameter identifiability. The result suggests that the mean-reversion timescale associated with the memory variable is unstable across the two formulations.

The flow-impact parameter $c$ behaves quite differently. In
\cite{HalperinItkinMarketron2} its calibrated value was approximately $0.9$,
whereas in the present experiment the optimizer yields $c=0.62824$. The
identifiability analysis in \cref{sec:identifiability} shows that vanilla option
prices contain relatively little information about the absolute level of $c$,
which is substantially less tightly identified than the principal smile
parameters. Nevertheless, conditional on the physical-measure block and the
other option-facing parameters used here, the SPX surface favors a lower
risk-neutral value of $c$ than the initial value inherited from the physical
measure. The fact that the optimum differs from this initial value suggests
that the vanilla surface contains some information about $c^Q$ beyond the
physical-measure specification, although the weak identification of $c$ implies
that this distinction should not be interpreted as strong evidence for a
precisely identified risk-neutral value. Finally, the comparison with
\cite{HalperinItkinMarketron2} shows that the two calibrated values of $c$ are
of the same order of magnitude despite the differences between the underlying
models.

The residual surface in \cref{fig:spx_iv_residuals_3d} and its projection in \cref{fig:spx_iv_residuals_2d} provide a further diagnostic of the calibration.
They show how the remaining pricing errors are distributed across strike and
maturity rather than only reporting their aggregate magnitude. In particular,
systematic residual patterns would indicate regions of the empirical SPX surface
that are not fully captured by the reduced model, whereas residuals that remain
centered around zero across the calibration domain would support the adequacy of
the model specification for the observed surface.

From a pricing perspective, the IV residual surface reveals a clear spatial
pattern in the model's fit across the observed domain of OTM calls and puts,
with moneyness $M \in [0.1,1.01]$. Residuals are closest to zero in a band
surrounding the approximate at-the-money region ($K \approx 2250$--$2260$),
indicating that the model fits near-ATM options most closely and that the
quality of the fit deteriorates toward the wings. The largest positive
residuals, corresponding to model over-prediction of IV and reaching
approximately $+0.0225$, concentrate at higher strikes and short maturities,
i.e., in the OTM-call region. The largest negative residuals, corresponding to
model under-prediction and reaching approximately $-0.040$, occur at lower
strikes across the maturity range, i.e., in the OTM-put region.

Consequently, the model provides a substantially better fit for near-ATM
options than for deep-OTM puts and, to a somewhat lesser extent, short-dated
OTM calls. The persistence of the short-maturity call residuals is notable
because the parameter $k$ is retained in the model specifically to control
short-time maturity behavior. The residual surface therefore suggests that
this mechanism, while important for the short-time asymptotics, is not
sufficient to capture the observed short-dated upside wing. The asymmetric
wing behavior points instead to two distinct directions for subsequent model
refinement: a better representation of the downside skew in OTM puts and an
additional mechanism for the short-term upside wing in OTM calls, precisely
where the absolute residuals are largest within the available market data.

\subsection{Another SPX experiment}

To examine the sensitivity of these conclusions to the calibration region, we
performed a second calibration with the same SPX data and the same exogenous
parameters, but with a wider admissible domain for the risk-neutral parameters.
In this experiment we used
\begin{equation}
(g,\bar y,\sigma_z,\rho,\bar f)=(0.3,1.5,0.8,-0.5,0.1), \qquad (\gamma,D)=(3,5),
\end{equation}
and initialized $(\sigma_0,\alpha,\beta,\mu,c,k)=(0.2,0,0.1,2,5,1)$. The bounds were enlarged to
\begin{equation}
\sigma_0\in[0.05,0.99],\quad \alpha\in[-2,2],\quad \beta\in[0.01,5], \quad
\mu\in[10^{-4},10],\quad c\in[0.01,15],\quad k\in[0.01,5].
\end{equation}
The resulting calibration was $(\sigma_0,\alpha,\beta,\mu,c,k) = (0.05004,\,0.34809,\,0.50984,\,0.00606,\,1.64334,\,0.86631)$, with $\operatorname{RMSE}_{\rm IV}=0.0165$.  Thus, the fit is somewhat better
than that of the baseline calibration, for which $\operatorname{RMSE}_{\rm
IV}=0.0184$, although the improvement should not be interpreted as a pure
optimizer comparison because the physical volatility parameter was changed from
$\sigma_z=0.5$ in the baseline experiment to $\sigma_z=0.8$ in the second
experiment.  More importantly, the two calibrations converge to materially
different parameter vectors, indicating the presence of multiple favorable
regions, or basins, in the parameter-to- option-price map. The computed residuals are also presented in \cref{fig:spx_iv_residuals_2}.

\begin{figure}[!htb]
\vspace*{-1em}
\centering
\subfloat[]{ \includegraphics[width=0.45\textwidth, height=7.cm] {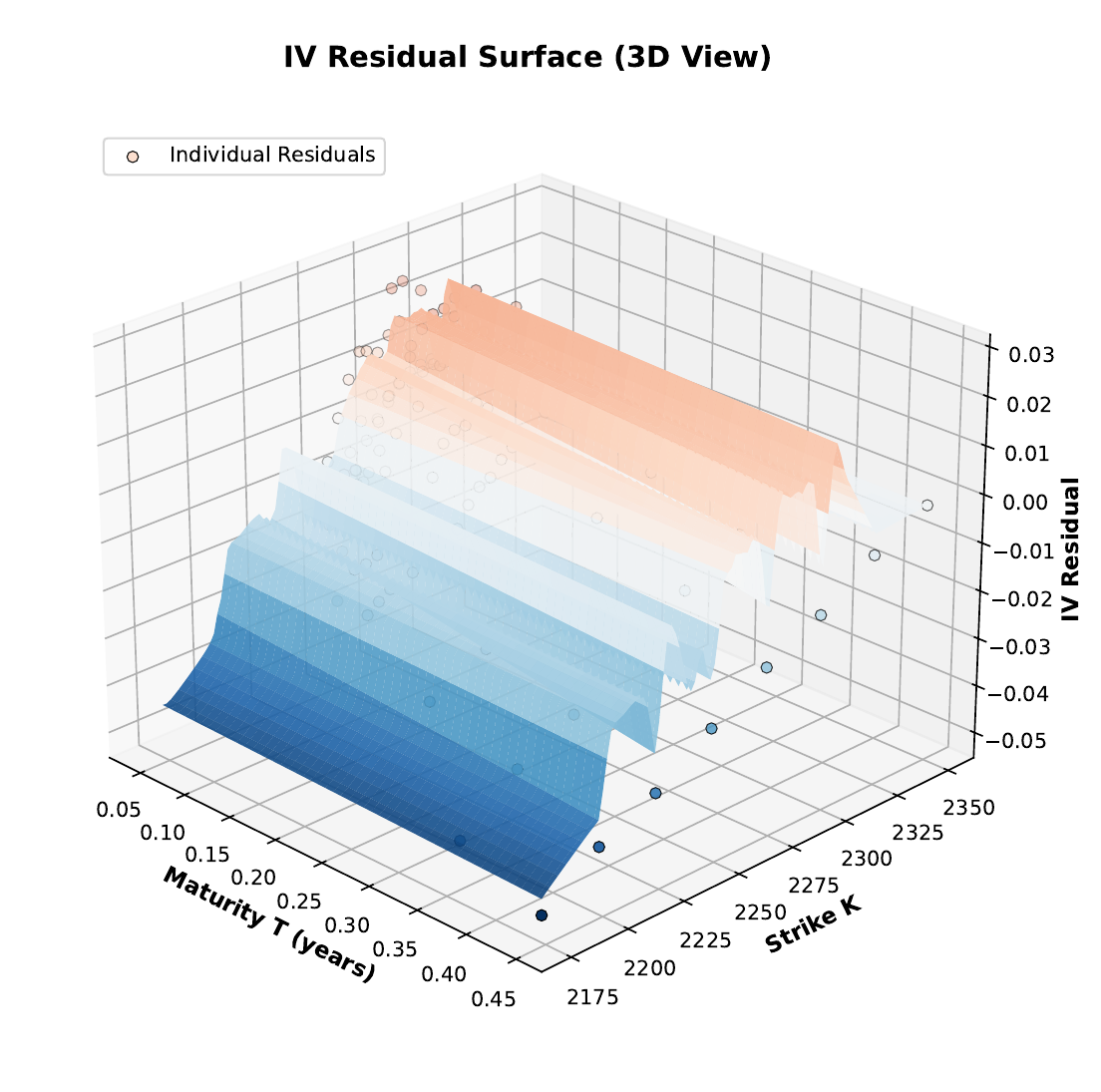} \label{fig:spx_iv_residuals_3d_2}}
\hfill
\subfloat[]{ \includegraphics[width=0.52\textwidth, height=7.cm] {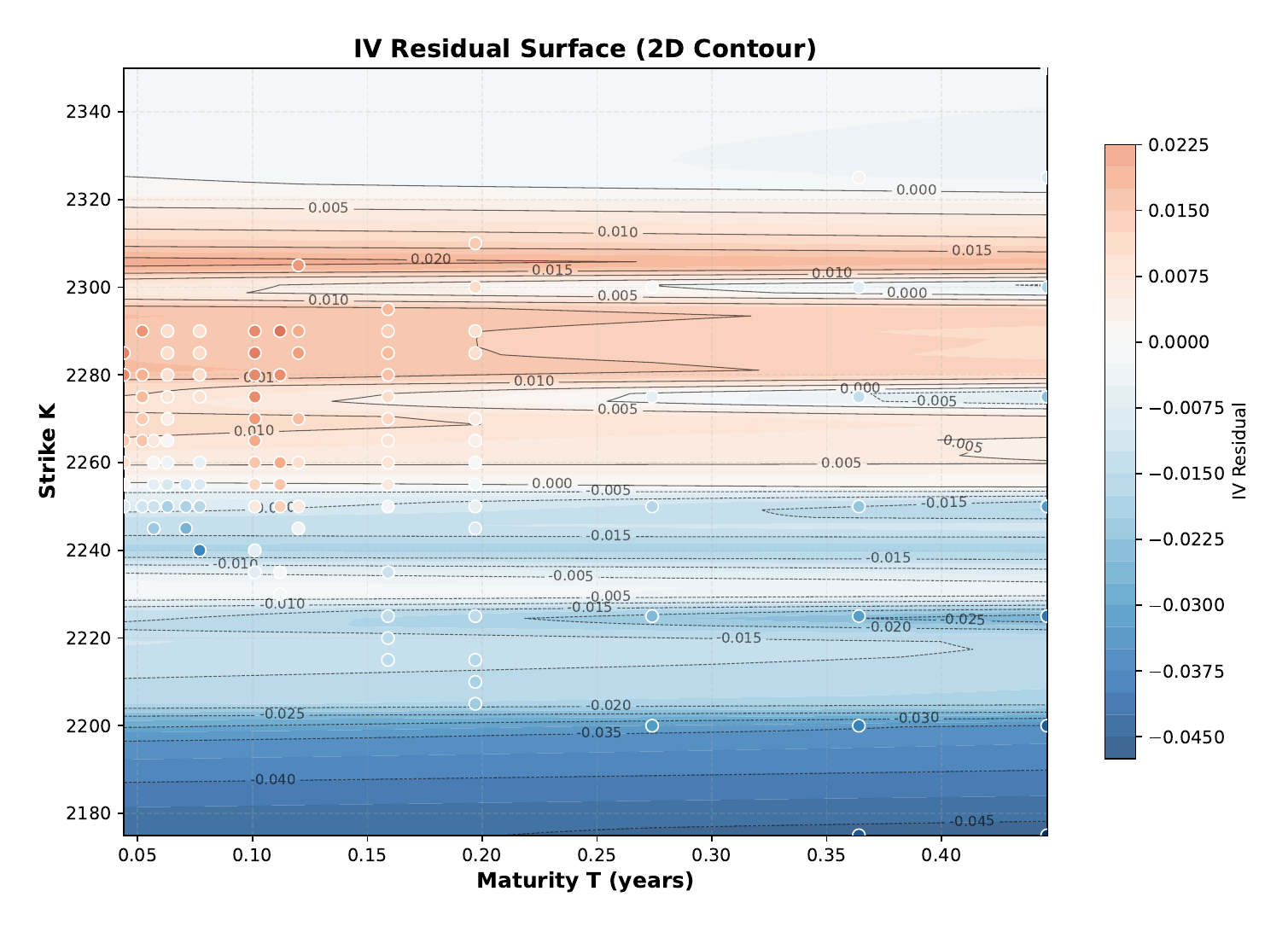} \label{fig:spx_iv_residuals_2d_2}}
\caption{SPX market data \textbf{second} experiment: implied-volatility
residual a) surface as a function of strike and maturity; b) 2D and calibration errors. The residuals are uniformly small across the calibration domain.}
\label{fig:spx_iv_residuals_2}
\end{figure}

The difference between the two solutions also illustrates the limited
identifiability of the parameters governing the state-dependent instantaneous
volatility.  In the baseline calibration, $(\alpha,\beta)=(2.0000,1.9235)$, with
$\alpha$ at its upper bound, whereas the wider calibration gives
$(\alpha,\beta)=(0.3481,0.5098)$, with both parameters well inside their bounds.
Since the option-pricing operator observes the volatility function $\sigma(y)$
rather than $\sigma_0$, $\alpha$, and $\beta$ separately, different combinations
of these parameters can generate similar effective volatility profiles over the
region of the memory state visited by the pricing dynamics.  The movement between
these two parameterizations therefore provides empirical evidence for the
approximate $(\sigma_0,\alpha,\beta)$ redundancy discussed above. The fact that
the two calibrations nevertheless produce similar overall pricing errors shows
that the option surface constrains the resulting volatility structure more
strongly than it constrains its particular parametric decomposition. This also
qualifies the local identifiability analysis of \cref{sec:identifiability}, whose
Hessian and profile widths are evaluated at the operating point
$\Theta_{\mathrm{red}}^{\star}$, which is far from either SPX optimum: the
$(\sigma_0,\alpha,\beta)$ direction that reads as stiff there is, at the fitted
surface, a shallow one along which $\alpha$ and $\beta$ trade off, much as $\rho$
does in \cref{sec:identifiability}.

The memory parameter $\mu$ provides a particularly clear result. In both
experiments the optimizer drives $\mu$ to the lower edge of the admissible
region: $\mu=0.01$ in the baseline calibration and $\mu=0.0061$ in the wider
calibration. Thus, the option data do not appear to support appreciable
mean-reversion of the memory variable through the explicit term
$\mu(\bar y-y)$. Instead, within the present reduced specification, the
effective stabilization of the memory dynamics appears to be generated
primarily through its coupling to the price and the nonlinear potential. This
conclusion should nevertheless be treated cautiously because the two
experiments use different $\sigma_z$ and therefore do not isolate the effect
of $\mu$ alone.

The flow-impact coefficient $c$ moves from $0.6282$ to $1.6433$ across the two
calibrations. This spread is expected rather than troubling: \cref{sec:identifiability}
identifies $c$ as the soft option-facing coordinate, so its point value is not
sharply pinned by a single surface. What both basins share is that $c^{\mathbb{Q}}$
does not collapse to the inherited physical value $c^{\mathbb{P}} = 0.9$, so the
staged calibration produces a nonzero $\mathbb{Q}$--$\mathbb{P}$ separation in the
flow block in either basin. The market price of flow risk is not this scalar
difference, however, but the wedge in the market price of risk $\lambda^{(x)}_t$
of \cref{sec:mpr}, into which $c$ enters together with $\bar f$ and the
state-dependent volatility. The two-basin spread should therefore be read as
evidence that a flow-risk premium is identifiable and nonzero, rather than as a
sharp estimate of its level.

By contrast, $k$ moves from $1.0028$ to $0.8663$, a considerably smaller relative change. This suggests that the signal relaxation timescale may be more stable
across calibrations, although the available data are not sufficient to claim
strong identification.

\subsection{The market price of risk} \label{sec:mpr}

In the original Marketron model of \cite{HalperinItkin2025Mark}, the three
Brownian motions $W^{(x)}_t$, $W^{(y)}_t$, and $W^{(\theta)}_t$ are mutually
independent. The utility-based market price of risk (MPR) is therefore a
three-component vector,
\begin{equation}
    \lambda_t = \left(\lambda^{(x)}_t,\ \lambda^{(y)}_t,\ \lambda^{(\theta)}_t\right),
\end{equation}
with one component per stochastic factor. The component associated with the
traded Brownian motion $W^{(x)}_t$ takes the same functional form as in a
complete market,
\begin{equation}
    \lambda^{(x)}_t = \frac{\mu_x(t, x_t, y_t, \theta_t) - r}{\sigma},
\end{equation}
while the components $\lambda^{(y)}_t$ and $\lambda^{(\theta)}_t$ correspond to
the unobservable memory and signal variables and are not directly measurable.

The reduced model changes this structure in two ways. First, adiabatic
elimination of the signal removes $\theta$ from the state space, so the third
component $\lambda^{(\theta)}_t$ disappears entirely. Since the reduced model
sets $\rho_{x\theta} = 0$ at leading order (the correlation enters only through
the first-order $1/k$ correction discussed in \cref{sec:transient}), there is
no residual signal contribution to the MPR:
\begin{equation}
    \lambda^{(\theta)}_t \equiv 0.
\end{equation}

Second, the introduction of the diffusive correlation $\rho$ between
$W^{(x)}_t$ and $W^{(y)}_t$ generates a non-zero second component
$\lambda^{(y)}_t$. Because the memory variable $y_t$ is now partially
correlated with the traded asset, it is partially hedgeable, and the market
demands compensation for its residual risk. The MPR vector in the reduced
model is therefore two-dimensional:
\begin{equation}
    \lambda_t = \left(\lambda^{(x)}_t,\ \lambda^{(y)}_t\right),
\end{equation}
where the first component follows from the indifference pricing framework and
takes the same functional form as in a complete market:
\begin{equation}
    \lambda^{(x)}_t = \frac{\bar\mu_x(t, x_t, y_t)}{\sigx(y_t)},
    \label{eq:mpr-reduced}
\end{equation}
with
\begin{equation}
    \bar\mu_x(t, x_t, y_t) = \barf(t) + \bar\eta(y_t) + \frac{\sigx^2(y_t)}{2}
    - r - c_x\, y_t \VM'(x_t),
\end{equation}
and $\bar\eta(y_t) = r - \sigx^2(y_t)/2$. The expression simplifies to
\begin{equation}
    \lambda^{(x)}_t = \frac{\barf(t) - c_x\, y_t \VM'(x_t)}{\sigx(y_t)}.
    \label{eq:mpr-simplified}
\end{equation}

The second component $\lambda^{(y)}_t$ is generated by the correlation $\rho$
through the market-price-of-risk terms in the pricing PDE \eqref{eq:pde2d}.
From the drift adjustment in the $y$-substep,
\begin{equation}
    \mu_y^{\mathbb{Q}} = \mu_y - \rho\,\frac{\bar\mu_x}{\sigx(y)},
\end{equation}
the component is
\begin{equation}
    \lambda^{(y)}_t = \frac{\rho\,\bar\mu_x(t, x_t, y_t)}{\sigx(y_t)}
    = \rho\,\lambda^{(x)}_t.
    \label{eq:mpr-y}
\end{equation}
Thus, the two components are proportional, with the correlation $\rho$ as the
proportionality constant. When $\rho = 0$, the MPR reduces to the single
component $\lambda^{(x)}_t$ of the uncorrelated model, consistent with
\cite{HalperinItkinMarketron2}. When $\rho \neq 0$, the market demands an
additional premium for bearing the partially hedgeable memory risk, and this
premium is exactly $\rho$ times the traded-asset premium.

In contrast to the complete-market setting, both components of the MPR in the
reduced Marketron are \emph{stochastic}: they depend on the memory variable
$y_t$ both directly through the term $c_x y_t \VM'(x_t)$ and indirectly through
the state-dependent volatility $\sigx(y_t)$. The memory variable encodes the
accumulated flow impact, and its fluctuations generate a time-varying,
state-dependent risk premium.

\myparagraph{Numerical computation.}
Since $y_t$ is unobservable, the empirical MPR is computed by Monte Carlo
simulation of the reduced system \eqref{eq:sde2d} under the physical measure.
Along each simulated path, we evaluate \eqref{eq:mpr-simplified} and
\eqref{eq:mpr-y} and take the expectation over the ensemble,
\begin{equation}
    \bar\lambda^{(x)}_t = \EE\Bigl[\lambda^{(x)}_t\Bigr]
    = \EE\Bigl[\frac{\barf(t) - c_x\, y_t \VM'(x_t)}{\sigx(y_t)}\Bigr],
    \qquad
    \bar\lambda^{(y)}_t = \rho\,\bar\lambda^{(x)}_t.
    \label{eq:mpr-expectation}
\end{equation}
Defaulted paths, where $x_t$ crosses the threshold $\hat x$ into the default
region, are excluded from the expectation, consistent with the treatment in
\cite{HalperinItkinMarketron2}. The standard error of the Monte Carlo estimate
is computed from the cross-sectional dispersion of the MPR across paths.

The wedge between the physical and risk-neutral measures is then quantified
through the \emph{level} of the MPR. The physical-measure MPR
$\bar\lambda^{(x),\mathbb{P}}_t$ is computed from the time-series calibration
of \cite{HalperinItkin2025Mark}, while the risk-neutral MPR
$\bar\lambda^{(x),\mathbb{Q}}_t$ is obtained from the option-implied
parameters. The difference,
\begin{equation} \label{eq:wedge}
\Delta\lambda_t \equiv \bar\lambda^{(x),\mathbb{Q}}_t - \bar\lambda^{(x),\mathbb{P}}_t,
\end{equation}
represents the market price of flow risk as a function of time.

\medskip\noindent\textit{Remark.} The MPR level and the wedge \eqref{eq:wedge}
were computed numerically for the full Marketron in \cite{HalperinItkinMarketron2},
and we do not reproduce those figures here. In isolation, with no second measure
or model to compare against, the reduced-model values would add length without
information. What the reduced core contributes is not a new number but the status
of the wedge: on the sloppy full model the $\mathbb{Q}$--$\mathbb{P}$ separation
is a ridge artifact, whereas on the reduced core it is a well-defined object,
subject to the weak identification of $c$ noted in \cref{sec:calibration}.

\section{Discussion} \label{sec:discussion}

The central message of this paper is epistemological rather than merely
computational.  The full Marketron, like many structural models in quantitative
finance, carries more parameters than the option surface can resolve.  This is
not a numerical inconvenience to be overcome with better optimizers.  It is a
structural defect that renders the model's economic claims unfalsifiable.  By
removing exact symmetries, freezing non-financial parameters, and adiabatically
eliminating the fast signal, we have reduced the model to an identifiable core
in which every parameter is either estimated from the data or fixed by an explicit
economic criterion.  The calibration objective is no longer a heuristic search
over a sloppy landscape but a far better-posed estimation problem, even though,
as the SPX fits show, individual coordinates such as $c$ and $\mu$ remain softly
determined and the optimum can sit at a bound.

This change of status has immediate economic consequences.  In the full
$18$-parameter model \eqref{eq:fullparams}, the flow-impact constant $c = c_x =
c_y$ sits on a continuous ridge along which it trades off against the memory
scale, the hidden volatility, and the initial state.  A reported value of $c$ is
therefore a point on a ridge, not an estimate.  On the reduced core
\eqref{eq:redparams}, the ridge is destroyed by the gauge fixing \eqref{eq:norm}
and the level absorption $\barh\to\bar y$, and $c$ becomes an identifiable coordinate
of the drift block, the softest of the option-facing set but no longer a ridge
direction.  The market price of flow risk, the wedge $\Delta\lambda_t$ of
\eqref{eq:wedge} between the physical and pricing values of the flow block, is
consequently a well-defined object rather than a point on a ridge: a compensation
for bearing non-hedgeable order-flow exposure that is identifiable rather than
asserted, though \cref{sec:calibration} shows its level is only weakly pinned by a
single surface. To our knowledge, this is the first such wedge made identifiable
within a structural model of inelastic
markets.

The correlation $\rho$ plays a subtle but essential role in this reduction. In
the uncorrelated model, risk aversion $\gamma$ and the hidden volatility
$\sigma_y$ collapse into the single product $p_y=\gamma$, removing one degree of
freedom from the diffusion structure.  The correlation $\rho$ restores it by
entering the pricing operator linearly through the cross-term $\partial^2_{xy}C$
and the drift adjustment, generating short-maturity skew that the state-dependent
drift alone cannot supply.  Financially, $\rho$ is the diffusive signature of
contemporaneous price impact; technically, it is the unique coefficient that
breaks the maturity degeneracy of \cref{prop:degenerate} at the short end.
Together with the effective constants $(\barf,\barh)$, which encode the
risk-adjusted drift wedge between the physical and pricing measures, $\rho$
completes the economic reading of the reduced surface. In the calibration itself $\rho$ is not free: it is fixed from the physical measure (\cref{sec:identifiability}), where it is directly estimable, because on the option surface its skew is confounded with $\alpha$ and $\beta$. Its role here is structural, restoring the degree of freedom that $\gamma$ removes, rather than serving as a calibrated coordinate.

The reduction is not without cost.  Adiabatic elimination of $\theta$ replaces
the rich signal-driven term structure of the full model with the deterministic
ramp \eqref{eq:ramp}, which captures only the relaxation of the conditional
mean.  When the signal autocorrelation is strong ($\sigma_z^2/2k\ll 1$) or when
option maturities span the relaxation time $1/k$ densely, the full
three-dimensional model with the transient correction of \cref{sec:transient}
should be preferred.  Similarly, the fluctuation-dissipation relation
\eqref{eq:fdt-marketron} reduces to the equality $k=\mu$ only at leading order;
for strongly nonlinear forcing, the higher harmonics of the trigonometric
signal introduce memory timescales that the single-mode kernel \eqref{eq:Kxx}
cannot resolve.  These are not failures of the reduced model but boundaries
of its domain of validity, and they are precisely the sort of falsifiable
conditions that an identifiable framework makes possible.

An important finding of the paper is that substituting the memory variable into the log-price equation yields a generalized Langevin equation with a closed-form, state-modulated memory kernel, while the memory variable itself provides an exact Markovian lift of this kernel. Thus, the representation in \cref{sec:gle} points to a broader class of models. The kernel \eqref{eq:Kxx} is separable and state-dependent, of the form $\bm K(Z_t,Z_s,\tau) = \boldsymbol\Phi(Z_t) \, \mathbf{k}(\tau) \, \boldsymbol\Phi(Z_s)^\top$ whereas, e.g., the framework of \cite{ItkinGLE1} assumes $\boldsymbol\Phi\equiv\mathbf{I}$. The smallest class containing both is therefore the separable state-dependent GLE. Four questions
must be settled before this broader class can be used as broadly as the lag-only
class:
\begin{itemize}
    \item Whether the Markovian embedding remains exact and finite-dimensional for a general $\boldsymbol\Phi$? In financial-mathematics terms (FMT), whether the corresponding non-Markovian stochastic Volterra equation with a state-dependent impact propagator admits an exact finite-dimensional Markovian realization?

    \item What replaces \eqref{eq:fdt} when the system--bath coupling is weighted by $\boldsymbol\Phi$, since the driving noise then becomes state-dependent and multiplicative? In FMT, what equilibrium-consistency condition replaces the usual relation between the impact-propagator decay and the exogenous-flow correlation time when the return innovation has state-dependent volatility?

    \item What the stationary law is, given that the Boltzmann form is already broken here, as \eqref{eq:Ueff} illustrates through the modulation-induced contribution to the effective potential? In FMT, what invariant distribution, if any, is induced by the state-dependent drift and diffusion of the reduced price--memory system, and how it differs from the equilibrium stationary distribution?

    \item Whether the triangular structure imposed on the kernel matrix by the absence of arbitrage survives state modulation, since the corresponding Girsanov density then acquires a state-dependent factor? In FMT, whether the no-arbitrage restriction on the impact-propagator matrix remains triangular under a state-dependent market price of memory risk, and whether the resulting change of measure remains compatible with the pricing dynamics?
\end{itemize}

A fifth question is empirical: the separable form is not identified without a normalization, because $\boldsymbol\Phi\mapsto\alpha\boldsymbol\Phi$ and
$\mathbf{k}\mapsto\alpha^{-2}\mathbf{k}$ leave the kernel unchanged, so the class
inherits a gauge freedom of exactly the kind \cref{sec:symmetries} removes here.

Two empirical directions are immediate. First, the model-implied VIX inherits
state dependence from the nonlinear indifference pricing operator and from the
default-region contribution of the potential, despite the absence of a separate
stochastic volatility factor, with instantaneous volatility given by the
state-dependent function $\sigx(y)$. Computing this index and comparing it with
the CBOE time series is a strict out-of-sample falsification test: if the two
regimes cannot lift the model VIX materially away from the level implied by
$\sigx(y)$, the Marketron in its present form cannot support volatility
derivatives without a stochastic spot-volatility extension.

Second, joint calibration to VIX options would constrain the soft $\mu$-direction
identified in \cref{sec:identifiability} and sharpen the premium estimates in
\eqref{eq:wedge}. Both tests are only meaningful on the identifiable core; on the
full model, a mismatch could always be attributed to a poor optimum.

The asymmetric wing behavior reported in \cref{sec:numerics} points to two
distinct directions for subsequent model refinement: a better representation of
the downside skew in OTM puts and an additional mechanism for the short-term
upside wing in OTM calls. In particular, introducing a separate stochastic
volatility factor could provide additional flexibility in the joint
strike--maturity dynamics of implied volatility and may help account for the
residual short-maturity call bias that is not eliminated by the current
$k$-dependent short-time mechanism. Such an extension would, however, need to
preserve the model's no-arbitrage structure and would introduce additional
parameters and a corresponding calibration trade-off. We reserve this extension for future investigation.

A further limitation is that the present calibration fixes the $P$-parameters
rather than estimating them jointly with the endogenous parameters. In a full
calibration, these parameters should themselves be inferred from the relevant
market data, and their estimation may alter both the fitted values and the
implied residual structure reported above. A natural direction for future
research is therefore to investigate the sensitivity of the model and its
pricing performance to both the $P$-parameters and the endogenous parameters,
including their joint calibration.

Finally, the computational cost of the reduced pricer could be further reduced by
constructing a neural surrogate trained on solutions of the reduced PDE
\eqref{eq:pde2d}. Such a surrogate could bring pricing times to the millisecond
range, making the model more suitable for live trading and intraday
recalibration. Because the reduced parameter space is low-dimensional and the
calibration problem exhibits substantially improved identifiability, the
training set could be generated on a regular lattice using standard supervised
learning methods. The surrogate would then approximate the pricing map of the
structurally constrained reduced pricer, while the latter would remain the
source of the model's arbitrage and boundary-condition guarantees. We also leave
this extension to future work.

Taken together with the analysis in \cite{ItkinGLE1}, the results of this
paper point to a broader interpretation of the Marketron as an economically
structured realization of generalized Langevin dynamics. The memory variable
in the reduced Marketron corresponds to a single exponential memory mode,
whereas the generalized Langevin formulation permits a general memory kernel
or a superposition of exponential modes, allowing the temporal structure of
memory to be determined empirically rather than imposed through a single
relaxation scale. Conversely, the Marketron provides an economic interpretation
of the memory and price-impact mechanisms that are introduced more
phenomenologically in the generalized Langevin framework. This connection
suggests a natural next step: to combine the economic structure of the
Marketron with the more general GLE memory specification and investigate its
implications in a joint calibration to SPX and VIX data, thereby testing
whether a common framework can consistently capture the short-, intermediate-,
and long-horizon dynamics of volatility and leverage.

\section*{Disclosure statement}

No potential conflict of interest was reported by the authors.

\section*{Funding}

No funding was received.

\section*{Disclaimer}

Opinions expressed here are author's own, and do not represent views of their employers. A standard disclaimer applies.

\section*{Acknowledgments}

I thank my long-time co-author Igor Halperin for our joint work on the Marketron model and valuable discussions, as well as Michael Isichenko for his constructive feedback regarding the Marketron model's dimensionality.

%%%%%%%%%%%%%%%%%%%%%%%%%%%%%%%%%%%%%%%%%%%%%%%%%%%%%%%%%%%%%%%%%%%%%%%%%%%%%

\printbibliography[title={References}]

\appendix

\appendixpage
\numberwithin{equation}{section}
\setcounter{equation}{0}

\section{Proof of \cref{prop:gauge}} \label{app:gauge}

\begin{proof}
Let $\tilde{y}_t = y_t / \lambda$ denote the transformed process under the gauge transformation \eqref{eq:gauge}. Applying It\^o's lemma to $\tilde{y}_t$, we obtain $d\tilde{y}_t = \frac{1}{\lambda} dy_t$. Under this mapping, the spatial parameters scale proportionally, meaning $\bar{y} \to \bar{y}/\lambda$, $y_0 \to y_0/\lambda$, and $\sigma_y \to \sigma_y/\lambda$, while the mean-reversion rate $\mu$ remains invariant. Consequently, the ratios $b_2/\sigma_y$, $\bar{y}/\sigma_y$, and $y_0/\sigma_y$ are invariant, as the scaling factor cancels out. The correlation parameters $\rho$ and $\rho_{x\theta}$ are dimensionless and therefore unaffected by the state scaling. Assuming the product $c_x c_y$ is preserved under the transformation, the pair $(x_t, \tilde{y}_t)$ exactly solves the stochastic system with the transformed parameters, leaving the law of the observable process $x_t$ unchanged.

To demonstrate that the option prices are preserved, we apply the spatial scaling to the indifference-pricing PDE \eqref{eq:pde}. By the chain rule, the partial derivatives with respect to the scaled variable transform as $\partial_y C \to \frac{1}{\lambda} \partial_y C$, $\partial^2_y C \to \frac{1}{\lambda^2} \partial^2_y C$, and $\partial^2_{xy} C \to \frac{1}{\lambda} \partial^2_{xy} C$. Substituting these alongside the transformed volatility $\sigma_y \to \lambda \sigma_y$, we evaluate the key PDE components:
\begin{align*}
\gamma (\lambda \sigma_y)^2 (1-\rho^2) \left( \frac{1}{\lambda} \partial_y C \right)^2 &= \gamma \sigma_y^2 (1-\rho^2) (\partial_y C)^2, \\
(\lambda \sigma_y)^2 \left( \frac{1}{\lambda^2} \partial^2_y C \right) &= \sigma_y^2 \partial^2_y C, \\
\rho \sigma (\lambda \sigma_y) \left( \frac{1}{\lambda} \partial^2_{xy} C \right) &= \rho \sigma \sigma_y \partial^2_{xy} C.
\end{align*}
Since these terms are homogeneous of degree zero under the joint transformation $y \to \lambda y$ and $\sigma_y \to \lambda \sigma_y$, the PDE \eqref{eq:pde} is strictly form-invariant. Because both the observable system dynamics and the pricing equation are invariant, it follows that every option price remains unchanged.

\end{proof}

\section{Fixing the regularizer} \label{app:eps}

The regularizer $\bar\epsilon$ is not a financial parameter. This appendix fixes it, and in doing so settles which criterion can do the job, since the obvious one cannot.

\myparagraph{Two admissible regularizers.}
Write $\VM'(x) = -e^{-x}\bigl[1 - R(x)\bigr]$. The original specification of \cite{HalperinItkin2025Mark} takes
\begin{equation} \label{eq:R1}
R_1(x) = \frac{g}{e^{x} + \bar\epsilon g},
\end{equation}
whereas the pricer of \cite{HalperinItkinMarketron2} takes
\begin{equation} \label{eq:R2}
R_2(x) = \frac{1}{2\bar\epsilon}\bigl[1 - \erf(x + b)\bigr],
\qquad
b = 4 - \erf^{-1}\!\Bigl[1 - \frac{2}{1 + e^{4} g \bar\epsilon}\Bigr],
\end{equation}
because the quadrature underlying the Gaussian RBF matrix elements closes in elementary functions for $R_2$ and does not for $R_1$. Both share the limits $R \to 0$ as $x \to +\infty$ and $R \to 1/\bar\epsilon$ as $x \to -\infty$, so both reduce the divergence of $\VM'$ in the default region from $e^{-2x}$ to $e^{-x}$, and in both
\begin{equation} \label{eq:epsasym}
\VM'(x) \;\sim\; \Bigl(\frac{1}{\bar\epsilon} - 1\Bigr) e^{-x}, \qquad x \to -\infty.
\end{equation}
We adopt $R_2$ throughout, so that the model as analyzed and the model as priced are the same object.

\myparagraph{The constant $b$ anchors $R_2$ to $R_1$.}
The definition of $b$ in \eqref{eq:R2} looks arbitrary but is not. Substituting $x = -4$ gives $\erf(b-4) = -\bigl[1 - 2/(1 + e^4 g\bar\epsilon)\bigr]$, whence
\begin{equation} \label{eq:anchor}
R_2(-4) \;=\; \frac{g e^{4}}{1 + \bar\epsilon g e^{4}} \;=\; R_1(-4).
\end{equation}
The two regularizers agree exactly at $x_\star = -4$, and $b$ is precisely the constant that makes them do so. This is a matching condition in the default region, which is where regularization is needed, and it is the reason the substitution is legitimate.

\myparagraph{The turning point cannot fix $\bar\epsilon$.}
Under $R_1$ the condition $\VM'(\hat x) = 0$ is linear in $\bar\epsilon$ once the denominator is cleared, giving $\bar\epsilon = 1 - e^{\hat x}/g$ and inviting the criterion that the mean-reversion range extend to a chosen $\hat x$. Under $R_2$ this fails. The turning point solves $\erf(\hat x + b) = 1 - 2\bar\epsilon$, that is
\begin{equation} \label{eq:xhatR2}
\hat x = \erf^{-1}(1 - 2\bar\epsilon) - b(g, \bar\epsilon),
\end{equation}
and because $b$ carries its own dependence on $\bar\epsilon$ the two contributions very nearly cancel. At $g = 1$, $\bar\epsilon = 0.1$ we find $\partial \hat x/\partial \log \bar\epsilon = -0.014$ against $\partial \hat x/\partial \log g = +0.389$, so the turning point is some twenty-five times more responsive to the coupling than to the regularizer, and it is not even monotone in $\bar\epsilon$. Over $\bar\epsilon \in (0.01, 0.45)$ and $g \in (0.1, 10)$ the reachable range of $\hat x$ is roughly $(-3.5, -1.6)$, so a target such as $\hat x = x_0 - 2$ is in general not attainable at all. Under $R_2$ the turning point is a property of $g$, not of $\bar\epsilon$, and the criterion must be sought elsewhere.

\myparagraph{The default rate cannot fix it either.}
Since the regularizer governs the deep region, it is natural to ask whether $\bar\epsilon$ can be calibrated to an observed default rate. It cannot, for two reasons. The first is the turning point again: with the barrier top drifting non-monotonically over $\hat x \in (-2.73, -2.38)$ as $D$ ranges over two orders of magnitude, a fixed default threshold lies sometimes inside the trapping region and sometimes outside it, and the resulting first-passage probability is not monotone in $\bar\epsilon$. In simulations of \eqref{eq:sde2d} at $\sigma = 1$ with a threshold at $x = -2.5$ we obtain one-year rates of $0.047\%$, $0.068\%$, $0.090\%$ and $0.045\%$ as $D$ increases through $50, 200, 800, 1500$, which admits no inversion. The second reason is structural and would remain even if the first were repaired: the regularizer exists in order to make the deep region unreachable, so it suppresses precisely the events one would calibrate against. At $\sigma = 0.6$ the same experiment returns a default rate indistinguishable from zero for every $D$ we tried. A quantity whose purpose is to prevent defaults is a poor instrument for setting their frequency.

\myparagraph{The criterion, and why it is imposed at $x_\star$.}
What $\bar\epsilon$ does control, monotonically and by construction, is the strength of the restoring force at depth, which is the quantity that must be bounded if the model is not to impose unrealistic drifts on distressed names. We therefore impose
\begin{equation} \label{eq:crit}
\bigl|\VM'(x_\star)\bigr| = D
\end{equation}
at a reference depth $x_\star$, with $D$ a chosen bound. Taking $x_\star = -4$ makes the criterion independent of which regularizer is used, by \eqref{eq:anchor}, and has the further advantage that the inversion remains closed form even though $R_2$ is transcendental. Writing $A = g e^{-x_\star}$ and $B = 1 + D e^{x_\star}$, condition \eqref{eq:crit} reads $A/(1 + \bar\epsilon A) = B$, so that $\bar\epsilon = 1/B - 1/A$, which is \eqref{eq:epsfix}. Note that a restoring force at depth requires $R(x_\star) > 1$, so $D$ is attained on the branch $\VM'(x_\star) > 0$. We have verified this against direct evaluation of $\VM'$ under $R_2$ to fourteen digits.

The admissible range follows from $\bar\epsilon > 0$, which is needed to keep $R_1$ pole-free and to keep $R_2$ decreasing, and reads
\begin{equation} \label{eq:gadmiss}
g \;>\; e^{x_\star}\bigl(1 + D e^{x_\star}\bigr).
\end{equation}
At $x_\star = -4$ this is a weak restriction: $g > 0.018$ for $D = 0$, $g > 0.153$ for $D = 400$, and $g > 0.354$ for $D = 1000$. Setting $D = 0$ recovers $\bar\epsilon = 1 - e^{x_\star}/g$, which is the turning-point formula with $\hat x = x_\star$; this is the one turning point that both regularizers share, and it is available only because it sits at the anchor.

\myparagraph{The potential and the sign structure.}
Integrating $R_2$ gives $\VM$ in closed form,
\begin{equation} \label{eq:VMclosed}
\VM(x) = \frac{1}{2\bar\epsilon}\Bigl\{ \erfc(b) + e^{-x}\bigl[\erf(b+x) - 1\bigr]
+ e^{\,b + 1/4}\bigl[\erf(b + \tfrac12) - \erf(b + x + \tfrac12)\bigr] \Bigr\},
\end{equation}
up to an additive constant, which is immaterial since only $\VM'$ enters the price drift and only differences of $\VM$ enter the effective potential \eqref{eq:Ueff}. For $x > \hat x$ the bracket $1 - R(x)$ is positive and $\VM' < 0$; for $x < \hat x$ it is negative and $\VM' > 0$. The turning point is therefore a maximum of $\VM$, the top of the barrier separating the default region from the normal one, rather than a resting point. Consequently the modulation $\VM'(x_t)\VM'(x_s)$ in the kernel \eqref{eq:Kxx} is negative when the two times sit on opposite sides of $\hat x$ and positive when they sit on the same side: memory reverses sign across the barrier, which is the sense in which the default region carries its own dynamics.

\section{Proof of \cref{prop:degenerate}} \label{app:degen}

We first establish that \eqref{eq:pde} is correctly derived from the indifference principle, ensuring the proposition addresses the model's fundamental properties rather than an artifact of the derivation. Subsequently, we demonstrate the implications of this equation for a contingent claim written exclusively on $x_T$.

\myparagraph{Validation of the derivation.}
Applying the transformation $u = -e^{-\gamma z}v$, the bracket maximized over the hedge $h$ forms a quadratic expression with a negative leading coefficient, yielding the maximizer $h^\star = v_s/(\gamma v) + \bar\mu_x/(\gamma s \sigma^2)$. By substituting $h^\star$ and dividing by $-e^{-\gamma z}$, the term $s\bar\mu_x v_s$ emerges twice with opposite signs—once originating from the uncontrolled generator and once from the maximized bracket—resulting in an exact cancellation. This reflects the standard financial argument that the drift of a traded asset is entirely hedged away, restricting the survival of $\bar\mu_x$ strictly to the quadratic source term. Transitioning to the certainty equivalent $C = \gamma^{-1}\log v$, the component $\tfrac12\sigma^2 s^2\gamma C_s^2$ induced by $v_{ss}$ cancels exactly against $-\tfrac12\gamma\sigma^2 s^2 C_s^2$ arising from $v_s^2/v$. Consequently, \eqref{eq:pde} features quadratic gradient terms exclusively in the hidden factors and none in $x$, confirming that the $x$-direction is spanned by the traded asset and thus carries no residual risk. Finally, applying the coordinate transformations $s\partial_s = \partial_x$ and $s^2\partial_{ss} = \partial_{xx} - \partial_x$ converts $\tfrac12\sigma^2 s^2 C_{ss}$ into $\tfrac12\sigma^2(C_{xx} - C_x)$, yielding the discounted drift coefficient $\bar\eta = r - \sigma^2/2$. The integrity of these cancellations and the final differential equation have been verified via symbolic algebra.

\myparagraph{Implications for an $x$-only payoff.}
Defining the pricing difference as $P = C - C^0$ and subtracting the respective PDEs, we observe that the source term $-\bar\mu_x^2/(2\gamma\sigma^2)$ is common to both equations and additive in neither $C$ nor $C^0$. Thus, it cancels identically irrespective of the terminal data. Expanding the difference of the quadratic terms yields
\begin{equation} \label{eq:quadremain}
(C_y)^2 - (C^0_y)^2 = 2 C^0_y P_y + P_y^2 ,
\end{equation}
with a symmetric expansion holding for the $\theta$ derivatives, such that
\begin{equation} \label{eq:Peq}
P_t + \tfrac12\sigma^2 P_{xx} + \bar\eta P_x - rP
+ \tfrac12\sigma_y^2 P_{yy} + \mu_y P_y + \gamma\sigma_y^2 C^0_y P_y + \tfrac12\gamma\sigma_y^2 P_y^2
+ (\theta\text{-terms}) = 0 ,
\end{equation}
subject to the terminal condition $P(T) = \calV(x_T)$. Crucially, every term capable of inducing dependence on the hidden factors is strictly proportional to $P_y$, $P_\theta$, or their corresponding derivatives. Because the terminal payoff $\calV$ depends solely on $x_T$, it follows that $P_y(T) = P_\theta(T) = 0$. Therefore, the trivial solution $P_y \equiv P_\theta \equiv 0$ is consistent with \eqref{eq:Peq}, and by uniqueness, it is the exact solution. Eliminating these terms reduces the system to $P_t + \tfrac12\sigma^2 P_{xx} + \bar\eta P_x - rP = 0$, recovering the classic Black-Scholes equation.

From a numerical perspective using a Strang splitting scheme, an equivalent argument applies. After applying the Cole-Hopf transformation $w = e^{\gamma C}$, the hidden-factor substep becomes linear in $w$. The terminal ratio $w/w^0 = e^{\gamma\calV(x)}$ is uniform in $y$ and $\theta$, allowing it to commute with any fractional step containing only hidden-factor derivatives, regardless of how those coefficients depend on $x$. Consequently, the ratio passes through the hidden-factor substep unchanged. The subsequent $x$-substep, featuring $y$-independent coefficients, structurally preserves this difference. Neither fractional step can introduce hidden-factor dependence into $P$.

\myparagraph{Scope and limitations of the proposition.}
It is important to emphasize that this proposition does not imply $C = C^0$, nor does it suggest that either certainty equivalent is independent of the memory parameters. On the contrary, both exhibit strong dependence on these parameters through the source term. In our numerical experiments using the parameters outlined in \cref{sec:numerics}, the gradient $C_y$ spans a range of $\pm 8$ across the computational grid, whereas the difference $C_y - C^0_y$ remains on the order of $10^{-15}$ within the interior domain. This indicates that the two certainty equivalents differ exclusively by a function of $(t,x)$. It is this shared $y$-profile (rather than a flat spatial gradient) that renders the final option price blind to the order flow. Furthermore, this degenerate behavior breaks down if $\sigma$ becomes a function of $y$, as the $x$-substep coefficients would then introduce $y$-dependence, making $P_y \equiv 0$ inconsistent. Similarly, the proposition ceases to hold if the terminal payoff explicitly incorporates $y_T$, which would force a non-zero terminal gradient $P_y(T) \neq 0$.

\section{Indifference pricing with correlated hidden factors} \label{app:corr}

This appendix derives \eqref{eq:pde}--\eqref{eq:residcov}. The uncorrelated case is in \cite{HalperinItkinMarketron2}; what follows carries $\rho$ and $\rho_{x\theta}$ through the same argument. Throughout, $s_t$ is the discounted price, $Z_t$ the discounted wealth, and the hidden factors are $h = (y,\theta)^\top$.

\myparagraph{Hamilton-Jacobi-Bellman equation.}
With the correlation structure \eqref{eq:corr}, the wealth increment $dZ_t = h_t\,ds_t$ now covaries with both hidden factors, so the generator acquires two cross terms in $z$ that are absent when the Brownian motions are independent. The value function $u(t,z,s,y,\theta)$ satisfies
\begin{align} \label{eq:hjb-corr}
0 = \frac{\partial u}{\partial t}
&+ \frac{1}{2}\sigma^2 s^2 \frac{\partial^2 u}{\partial s^2}
+ \frac{1}{2}\sigma_y^2 \frac{\partial^2 u}{\partial y^2}
+ \frac{1}{2}\sigma_\theta^2 \frac{\partial^2 u}{\partial \theta^2}
+ \rho\,\sigma\sigma_y s \frac{\partial^2 u}{\partial s \partial y}
+ \rho_{x\theta}\,\sigma\sigma_\theta s \frac{\partial^2 u}{\partial s \partial \theta}
\nonumber \\
&+ s\bar\mu_x \frac{\partial u}{\partial s}
+ \mu_y \frac{\partial u}{\partial y}
+ \mu_\theta \frac{\partial u}{\partial \theta}
+ \max_{h}\Bigl[ \tfrac{1}{2}h^2 s^2\sigma^2 \frac{\partial^2 u}{\partial z^2}
+ h\,\mathcal{B}[u] \Bigr],
\end{align}
where
\begin{equation} \label{eq:Bop}
\mathcal{B}[u] \;=\; s^2\sigma^2 \frac{\partial^2 u}{\partial z \partial s}
\;+\; s\sigma\sigma_y\rho\,\frac{\partial^2 u}{\partial z \partial y}
\;+\; s\sigma\sigma_\theta\rho_{x\theta}\,\frac{\partial^2 u}{\partial z \partial \theta}
\;+\; s\bar\mu_x \frac{\partial u}{\partial z}.
\end{equation}
The last two terms of \eqref{eq:Bop} beyond the first and fourth are new; setting $\rho = \rho_{x\theta} = 0$ recovers the bracket of \cite{HalperinItkinMarketron2}.

\myparagraph{Optimal hedge.}
The bracket in \eqref{eq:hjb-corr} is quadratic in $h$ with negative leading coefficient, so the maximizer is $h^\star = -\mathcal{B}[u]/(s^2\sigma^2 u_{zz})$ and the maximized value is $-\mathcal{B}[u]^2/(2s^2\sigma^2 u_{zz})$. Exponential utility factorizes the value function as $u = -e^{-\gamma z}v(t,s,y,\theta)$, so every $z$-derivative brings down a factor of $\gamma$ and the exponential cancels, giving
\begin{equation} \label{eq:hstar}
h^\star \;=\; \frac{1}{\gamma}\frac{\partial_s v}{v}
\;+\; \frac{\rho\,\sigma_y}{\gamma s\sigma}\frac{\partial_y v}{v}
\;+\; \frac{\rho_{x\theta}\,\sigma_\theta}{\gamma s\sigma}\frac{\partial_\theta v}{v}
\;+\; \frac{\bar\mu_x}{\gamma s\sigma^2}.
\end{equation}
In terms of the certainty equivalent $C = \gamma^{-1}\log v$ and the log-price, this reads
\begin{equation} \label{eq:hstarC}
s\,h^\star \;=\; \underbrace{\frac{\partial C}{\partial x}}_{\text{delta}}
\;+\; \underbrace{\frac{1}{\sigma}\Bigl(\rho\,\sigma_y \frac{\partial C}{\partial y}
+ \rho_{x\theta}\,\sigma_\theta \frac{\partial C}{\partial \theta}\Bigr)}_{\text{correlation hedge}}
\;+\; \underbrace{\frac{\bar\mu_x}{\gamma\sigma^2}}_{\text{Merton}} .
\end{equation}
The middle term is the whole effect of the correlations on the hedge: the traded asset is used to hedge the part of the hidden noise that it spans. It vanishes with $\rho$ and $\rho_{x\theta}$, and \eqref{eq:hstarC} then reduces to delta plus Merton.

\myparagraph{Reduction to the certainty equivalent.}
Substituting $h^\star$ back and writing $q \equiv \rho\,\sigma_y \partial_y C + \rho_{x\theta}\,\sigma_\theta \partial_\theta C$, the terms in the resulting equation for $v$ that are quadratic in derivatives organize as $-(P + \gamma v q + R)^2/(2v)$ with $P = \sigma\,\partial_x v$ and $R = \bar\mu_x v/\sigma$. Two cancellations then occur when $v = e^{\gamma C}$ is inserted.

The first is the classical Cole-Hopf cancellation. The term $\tfrac12\sigma^2 \partial_x^2 v$ contributes $\tfrac{\gamma}{2}\sigma^2 (\partial_x C)^2$, which is removed exactly by $-P^2/(2v)$. The second is new and is what makes the correlated problem tractable: the cross contributions $\rho\sigma\sigma_y\partial^2_{xy}v$ and $\rho_{x\theta}\sigma\sigma_\theta\partial^2_{x\theta}v$ generate $\gamma\sigma\,\partial_x C\,q$, which is removed exactly by the cross product $-2Pq\gamma v/(2v)$. No term mixing $\partial_x C$ with the hidden gradients survives.

What remains of the quadratic terms is
\begin{equation} \label{eq:quadres}
\frac{\gamma}{2}\Bigl[\sigma_y^2 (\partial_y C)^2 + \sigma_\theta^2 (\partial_\theta C)^2 - q^2\Bigr]
\;=\; \frac{\gamma}{2}\,\nabla_h C^\top \widetilde\Sigma\, \nabla_h C ,
\end{equation}
with $\widetilde\Sigma$ as in \eqref{eq:residcov}, since expanding $q^2$ produces $\rho^2\sigma_y^2(\partial_y C)^2 + \rho_{x\theta}^2\sigma_\theta^2(\partial_\theta C)^2 + 2\rho\rho_{x\theta}\sigma_y\sigma_\theta \partial_y C\,\partial_\theta C$. The cross product of $q$ with $R$ supplies the two drift adjustments $-\rho\sigma_y\bar\mu_x/\sigma$ and $-\rho_{x\theta}\sigma_\theta\bar\mu_x/\sigma$ appearing in \eqref{eq:pde}, and $R^2$ supplies the term $-\bar\mu_x^2/(2\gamma\sigma^2)$ unchanged. The drift and discount terms in $x$ are unaffected by the correlations and carry over from \cite{HalperinItkinMarketron2}. This establishes \eqref{eq:pde}.

\myparagraph{The residual covariance is a Schur complement.}
Write the full instantaneous covariance of $(x, h)$ in block form, with $\Sigma_{hh} = \mathrm{diag}(\sigma_y^2, \sigma_\theta^2)$ because $\rho_{y\theta} = 0$, and $\Sigma_{hx} = (\rho\sigma\sigma_y,\ \rho_{x\theta}\sigma\sigma_\theta)^\top$. Then
\begin{equation} \label{eq:schur}
\widetilde\Sigma \;=\; \Sigma_{hh} \;-\; \frac{1}{\sigma^2}\,\Sigma_{hx}\Sigma_{hx}^\top ,
\end{equation}
which is precisely \eqref{eq:residcov}. The interpretation is immediate: $\widetilde\Sigma$ is the conditional covariance of the hidden factors given the traded one, so the quadratic penalty in \eqref{eq:pde} charges the investor only for the hidden risk that the traded asset cannot span. The off-diagonal entry is negative whenever $\rho$ and $\rho_{x\theta}$ share a sign, because hedging through the common traded asset induces a negative residual dependence between the two hidden factors even though their driving noises are uncorrelated.

\myparagraph{Limiting cases.}
At $\rho = \rho_{x\theta} = 0$ we have $q \equiv 0$, so \eqref{eq:hstarC} loses its middle term, the drift adjustments vanish, and $\widetilde\Sigma = \mathrm{diag}(\sigma_y^2,\sigma_\theta^2)$, which returns the quadratic gradient terms $\tfrac{\gamma}{2}\sigma_y^2(\partial_y C)^2 + \tfrac{\gamma}{2}\sigma_\theta^2(\partial_\theta C)^2$ of \cite{HalperinItkinMarketron2}.

In the scalar case, with $\theta$ absent, $\widetilde\Sigma$ collapses to $\sigma_y^2(1-\rho^2)$, so the quadratic term is $\tfrac{1}{2}\gamma(1-\rho^2)\sigma_y^2(\partial_y C)^2$ and the effective distortion exponent is $\gamma(1-\rho^2)$, in agreement with \cite{HH2002, MusielaZariphopoulou2004}. This is the origin of $p_y = \gamma(1-\rho^2)$ in the reduced parameter list: the derived combination, not $\gamma$ itself, is what the pricing equation sees.

\myparagraph{The reduced equation.}
Adiabatic elimination of $\theta$ replaces $f(\theta)$ and $h(\theta)$ by their stationary averages and removes $\theta$ from the state. In \eqref{eq:pde} this deletes the row and column of $\widetilde\Sigma$ carrying $\theta$, so the quadratic term becomes $\tfrac{1}{2}\gamma(1-\rho^2)\sigma_y^2(\partial_y C)^2$, deletes the $\theta$ drift and diffusion terms, and leaves the $x$ and $y$ blocks untouched. The result is \eqref{eq:pde2d}. Note that $\rho_{x\theta}$ does not disappear silently: it enters the reduced description only through the transient correction of \cref{sec:transient}, at order $1/k$, which is why the baseline reduced model sets it to zero.

\section{Proof of Proposition~\ref{prop:marketron-gle}} \label{app:gle}

The second line of \eqref{eq:sde2d-again} is linear in $y$ with constant coefficients, so it integrates in closed form,
\begin{equation} \label{eq:ysol}
y_t \;=\; e^{-\mu t}y_0
\;+\; \hat y\bigl(1-e^{-\mu t}\bigr)
\;-\; c_y\int_0^t e^{-\mu(t-s)}\VM(x_s)\,ds
\;+\; \sigma_y\int_0^t e^{-\mu(t-s)}dW^{(y)}_s ,
\end{equation}
with $\hat y=\bar y+\barh/\mu$ the level to which $y$ relaxes in the absence of coupling. The memory variable is thus an exponentially weighted average of $\VM(x_s)$ along the realized path, with forgetting rate $\mu$.

Substituting \eqref{eq:ysol} into the drift of $x$ gives
\begin{equation} \label{eq:xdrift}
\barf + \bar\eta - c_x\,\hat y\,\VM'(x_t)
\;+\; c_x c_y\,\VM'(x_t)\!\int_0^t\! e^{-\mu(t-s)}\VM(x_s)\,ds
\;-\; c_x\,\sigma_y \VM'(x_t)\!\int_0^t\! e^{-\mu(t-s)}dW^{(y)}_s,
\end{equation}
where terms carrying $e^{-\mu t}$ have been suppressed; they decay on the scale $1/\mu$ and do not affect the stationary structure. The first two terms are local in the state and supply the first two terms of \eqref{eq:Ueff}. The last is the multiplicative noise named in the proposition. The memory sits in the third term, which acts on the level of $\VM(x_s)$ rather than on a velocity and so is not yet in GLE form.

Integrate that term by parts. With $u=\VM(x_s)$ and $dv=e^{-\mu(t-s)}ds$, so that $v=e^{-\mu(t-s)}/\mu$,
\begin{equation} \label{eq:byparts}
\int_0^t e^{-\mu(t-s)}\VM(x_s)\,ds
\;=\; \frac{\VM(x_t)}{\mu}
\;-\; \frac{\VM(x_0)e^{-\mu t}}{\mu}
\;-\; \frac{1}{\mu}\int_0^t e^{-\mu(t-s)}\,d\VM(x_s).
\end{equation}
By It\^o's lemma $d\VM(x_s)=\VM'(x_s)\,dx_s+\tfrac12\VM''(x_s)\sigma^2 ds$, so
\begin{equation} \label{eq:byparts2}
\int_0^t e^{-\mu(t-s)}\VM(x_s)\,ds
= \frac{\VM(x_t)}{\mu}
- \frac{1}{\mu}\int_0^t e^{-\mu(t-s)}\VM'(x_s)\,\dot x_s\,ds
- \frac{\sigma^2}{2\mu}\int_0^t e^{-\mu(t-s)}\VM''(x_s)\,ds ,
\end{equation}
again dropping the transient. Multiplying by $c_x c_y\VM'(x_t)$, the first term on the right gives
\begin{equation}
\frac{c_x c_y}{\mu}\,\VM'(x_t)\VM(x_t)
\;=\; \frac{c_x c_y}{2\mu}\,\frac{d}{dx}\bigl[\VM(x)^2\bigr]\Big|_{x=x_t},
\end{equation}
which is the third term of \eqref{eq:Ueff} and is again local. The second term gives
\begin{equation} \label{eq:friction}
-\,\frac{c_x c_y}{\mu}\,\VM'(x_t)\int_0^t e^{-\mu(t-s)}\VM'(x_s)\,\dot x_s\,ds ,
\end{equation}
which is the friction term of an overdamped GLE with the tail part of \eqref{eq:Kxx}. The third term of \eqref{eq:byparts2} contributes
\begin{equation} \label{eq:itoterm}
-\,\frac{c_x c_y\sigma^2}{2\mu}\,\VM'(x_t)\int_0^t e^{-\mu(t-s)}\VM''(x_s)\,ds ,
\end{equation}
a further memory term acting on the level of $\VM''$. It is $O(\sigma^2)$, vanishes identically when $\VM$ is linear, and we absorb it into $\eta_t$; it does not affect the structure of the kernel.

Collecting the local contributions into $U_{\mathrm{eff}}$ and the noise sources into $\eta_t$, the drift equation for $x$ reads
\begin{equation} \label{eq:precollect}
\dot x_t \;+\; \frac{c_x c_y}{\mu}\VM'(x_t)\!\int_0^t\! e^{-\mu(t-s)}\VM'(x_s)\,\dot x_s\,ds
\;=\; -\,U_{\mathrm{eff}}'(x_t) \;+\; \eta_t .
\end{equation}
Writing the bare term as $\dot x_t=\int_0^t\delta(t-s)\,\dot x_s\,ds$ collects the left-hand side into a single convolution with the kernel $K_M$ of \eqref{eq:Kxx}, which is \eqref{eq:marketron-gle}. \qed

\section{RBF matrix elements of the reduced model} \label{app:rbf}

Appendix~A of \cite{HalperinItkinMarketron2} evaluates the matrix elements of the Volterra kernel in closed form. Its derivation rests on the Green's function being a product of three independent one-dimensional Gaussians, which factorizes every integral into three scalar ones with
\begin{equation} \label{eq:scalar-a}
\int_{-\infty}^{\infty} \frac{e^{-\varepsilon(\xi - x_k)^2 - (x-\xi)^2/(2\sigma^2\Delta\tau)}}{\sigma\sqrt{2\pi\Delta\tau}}\,d\xi
= \frac{1}{a(\sigma)}\,e^{-\varepsilon(x-x_k)^2/a^2(\sigma)},
\qquad a^2(\sigma) = 1 + 2\varepsilon\sigma^2\Delta\tau .
\end{equation}
Under the reduction the $\theta$ direction disappears, but the surviving pair $(x,y)$ is correlated through $\rho$, so the Green's function no longer factorizes and \eqref{eq:scalar-a} cannot be applied coordinate by coordinate.

One route is to rotate to the eigenbasis of the covariance. This works, because the Gaussian RBF is isotropic and therefore invariant under orthogonal transformations, but it is not necessary: the entire appendix generalizes by promoting the scalar $a^2(\sigma)$ to a matrix, and the resulting formulas contain the uncorrelated ones as the diagonal case.

\myparagraph{The master identity.}
Write $\bm p = (\xi,\eta)^\top$ for the integration variable, $\bm u = (x,y)^\top$ for the evaluation point, $\bm p_k = (x_k, y_j)^\top$ for an RBF centre, and
\begin{equation} \label{eq:Sig2d}
\Sigma \;=\; \begin{pmatrix} \sigma^2 & \rho\,\sigma\sigma_y \\ \rho\,\sigma\sigma_y & \sigma_y^2 \end{pmatrix}
\end{equation}
for the diffusion covariance of the reduced model. With \eqref{eq:sigy} in force $\Sigma = \Sigma(y)$, since $\sigx$ appears in three of its four entries; the formulas below hold verbatim with $\Sigma$ evaluated at the $y$-node of the substep, at the cost of one $2\times2$ inverse per node rather than one overall. The Green's function of the homogeneous part of \eqref{eq:pde2d} is the bivariate Gaussian with covariance $\Sigma\Delta\tau$. Then, for the isotropic Gaussian RBF $e^{-\varepsilon|\bm p - \bm p_k|^2}$,
\begin{equation} \label{eq:master}
\int_{\mathbb{R}^2} e^{-\varepsilon|\bm p - \bm p_k|^2}\,G(\Delta\tau; \bm u \,|\, \bm p)\,d\bm p
\;=\; \frac{1}{\sqrt{\det \bm A}}\,
\exp\Bigl[-\varepsilon\,(\bm u - \bm p_k)^\top \bm A^{-1} (\bm u - \bm p_k)\Bigr],
\qquad
\bm A \;\equiv\; \bm I + 2\varepsilon\,\Sigma\,\Delta\tau .
\end{equation}
This follows from the convolution of two Gaussians, the RBF being an unnormalized Gaussian of covariance $(2\varepsilon)^{-1}\bm I$. When $\rho = 0$ the matrix $\bm A$ is diagonal with entries $a^2(\sigma)$ and $a^2(\sigma_y)$, and \eqref{eq:master} splits into the product of two copies of \eqref{eq:scalar-a}. The reduction therefore amounts to the substitution
\begin{equation} \label{eq:substitution}
\prod_{\text{coords}} \frac{1}{a(\cdot)}\,e^{-\varepsilon(\cdot)^2/a^2(\cdot)}
\;\longmapsto\;
\frac{1}{\sqrt{\det\bm A}}\,e^{-\varepsilon\,(\cdot)^\top \bm A^{-1}(\cdot)},
\end{equation}
applied wherever Appendix~A of \cite{HalperinItkinMarketron2} produces a product of $a$-factors, together with the deletion of the $a(\sigma_\theta)$ factor.

\myparagraph{Linear terms.}
The terms $\calA_1$ and $\calA_2$ of \cite{HalperinItkinMarketron2} carry a linear prefactor $(\bm p - \bm p_k)$ multiplying the RBF, produced by differentiating the Gaussian. Since $\partial_{\bm p_k} e^{-\varepsilon|\bm p - \bm p_k|^2} = 2\varepsilon(\bm p - \bm p_k)e^{-\varepsilon|\bm p - \bm p_k|^2}$, these integrals follow from \eqref{eq:master} by differentiation with respect to the centre,
\begin{equation} \label{eq:linear}
\int_{\mathbb{R}^2} (\bm p - \bm p_k)\,e^{-\varepsilon|\bm p - \bm p_k|^2}\,G(\Delta\tau; \bm u\,|\,\bm p)\,d\bm p
\;=\; \bm A^{-1}(\bm u - \bm p_k)\;\times\;
\frac{e^{-\varepsilon (\bm u - \bm p_k)^\top \bm A^{-1}(\bm u - \bm p_k)}}{\sqrt{\det \bm A}} .
\end{equation}
In the uncorrelated case $\bm A^{-1}$ is diagonal and each component reduces to the corresponding scalar result. With $\rho \neq 0$ the off-diagonal entries of $\bm A^{-1}$ mix the two directions, so the $x$-component of \eqref{eq:linear} acquires a contribution proportional to $(y - y_j)$ and conversely. This mixing is the only structural change: no new special functions appear and everything remains closed form.

\myparagraph{Quadratic terms.}
The term $\bar\Psi_2$ contains products of two RBFs centred at $\bm p_k$ and $\bm p_{k^*}$. Because the RBF is isotropic, such a product is again an isotropic Gaussian,
\begin{equation} \label{eq:rbfprod}
e^{-\varepsilon|\bm p - \bm p_k|^2}\,e^{-\varepsilon|\bm p - \bm p_{k^*}|^2}
\;=\; e^{-\tfrac{\varepsilon}{2}|\bm p_k - \bm p_{k^*}|^2}\;
e^{-2\varepsilon\bigl|\bm p - \bar{\bm p}\bigr|^2},
\qquad \bar{\bm p} = \tfrac{1}{2}(\bm p_k + \bm p_{k^*}),
\end{equation}
so the quadratic matrix elements are given by \eqref{eq:master} with $\varepsilon$ replaced by $2\varepsilon$, the centre replaced by the midpoint $\bar{\bm p}$, and an overall factor $\exp[-\tfrac{\varepsilon}{2}|\bm p_k - \bm p_{k^*}|^2]$. Explicitly, with $\bm A_2 \equiv \bm I + 4\varepsilon\Sigma\Delta\tau$,
\begin{equation} \label{eq:quadelem}
\int_{\mathbb{R}^2} e^{-\varepsilon|\bm p - \bm p_k|^2} e^{-\varepsilon|\bm p - \bm p_{k^*}|^2} G(\Delta\tau; \bm u\,|\,\bm p)\,d\bm p
\;=\; \frac{e^{-\tfrac{\varepsilon}{2}|\bm p_k - \bm p_{k^*}|^2}}{\sqrt{\det \bm A_2}}\,
\exp\Bigl[-2\varepsilon(\bm u - \bar{\bm p})^\top \bm A_2^{-1}(\bm u - \bar{\bm p})\Bigr].
\end{equation}
Linear prefactors in the quadratic block are handled by differentiating \eqref{eq:quadelem} with respect to $\bar{\bm p}$, exactly as in \eqref{eq:linear}.

\myparagraph{Distinguishing the covariance matrices.}
It is critical to distinguish between the two contexts in which a covariance matrix enters the formulation, as they represent distinct mathematical objects. The Green's function, and consequently the integral representations \eqref{eq:master}--\eqref{eq:quadelem}, incorporates the primary diffusion covariance $\Sigma$ defined in \eqref{eq:Sig2d}. Conversely, the quadratic gradient term within the pricing equation is governed by the residual covariance $\widetilde\Sigma$ of \eqref{eq:residcov}, which simplifies in the reduced model to the scalar $\sigma_y^2(1-\rho^2)$. As a result, the coefficient scaling $\calC_\eta^2$ inside $\bar\Phi$ becomes $\tfrac{1}{2}\gamma(1-\rho^2)\sigma_y^2 = \tfrac{1}{2}p_y\sigma_y^2$, where $p_y$ is the derived parameter introduced in \cref{sec:gamma}. Meanwhile, the kernel of the integral operator retains $\sigma_y^2$ alongside the full correlation structure. Conflating these two matrices would erroneously rescale the Green's function instead of the intended nonlinearity.

\myparagraph{Memory-dependent volatility.}
The introduction of \eqref{eq:sigy} preserves the preceding structural arguments. During the $x$-substep of the Strang splitting scheme, $y$ acts as a frozen parameter. Consequently, the differential operator $\tfrac12\sigx^2(y)\partial_{xx} + \bar\eta(y)\partial_x - r$ exhibits constant coefficients with respect to $x$, yielding a Gaussian Green's function with variance $\sigx^2(y)\Delta\tau$ and mean shift $\bar\eta(y)\Delta\tau$. The corresponding one-dimensional element is therefore evaluated as
\begin{equation} \label{eq:rbf-sigy}
\int_{\mathbb{R}} e^{-\varepsilon(\xi - x_k)^2} G(\Delta\tau; x\,|\,\xi)\,d\xi
= \frac{e^{-r\Delta\tau}}{a(\sigx(y))}
\exp\Bigl[-\frac{\varepsilon\,\bigl(x - x_k - \bar\eta(y)\Delta\tau\bigr)^2}{a^2(\sigx(y))}\Bigr],
\qquad a^2 = 1 + 2\varepsilon\sigx^2(y)\Delta\tau,
\end{equation}
a result we have verified against direct quadrature to machine precision. Relative to the constant-volatility regime, two computational adjustments are necessary: the scaling factor $a$ becomes an array indexed by the $y$-nodes rather than remaining a scalar, and the mean shift $\bar\eta(y)\Delta\tau$—previously an absorbable constant—must be explicitly tracked per node. Notably, neither modification compromises the closed-form tractability of the integrals. This highlights a fundamental advantage of the operator splitting framework: if applied directly to the unsplit operator, a $y$-dependent diffusion coefficient would fundamentally destroy the Gaussian structure of the Green's function.

\myparagraph{Summary of the model reduction.}
To adapt the framework presented in Appendix~A of \cite{HalperinItkinMarketron2} to the present reduced model, the following systematic modifications are required: first, eliminate the $\theta$ integral and the associated scaling factor $a(\sigma_\theta)$ entirely. Second, replace each surviving product of $a$-factors with the matrix formulation \eqref{eq:substitution}, setting $\bm A = \bm I + 2\varepsilon\Sigma\Delta\tau$ for the linear terms, or $\bm A_2 = \bm I + 4\varepsilon\Sigma\Delta\tau$ within the quadratic block. Third, substitute each scalar linear-prefactor result with the generalized form in \eqref{eq:linear}. Finally, fix the coefficient of the quadratic gradient term at $\tfrac{1}{2}p_y\sigma_y^2$. Under these modifications, every matrix element remains analytically tractable in closed form. The computational complexity per element is strictly preserved, with the sole addition being the inversion of a $2\times2$ matrix, an operation performed only once per time step as it is uniformly applicable across all elements.

The final reduced model retains nine parameters for the whole surface. The physical-measure parameters $\bar f$ and $c$ are inherited from the time-series calibration, except that $c$ is deliberately released in the option calibration. This is a compromise between statistical identifiability and economic identification: the eigenanalysis shows that $\bar f$ is a stiff direction whereas $c$ is substantially softer, but fixing $c$ at its physical-measure value would eliminate the option-implied flow-risk premium altogether. We therefore retain $c$ as a Q-measure calibration parameter and explicitly report its weak identification through profile likelihoods and Hessian eigenvalues. The difference between its Q- and P-measure estimates is consequently interpreted as a market price of flow risk, with the associated uncertainty carried into the economic interpretation.

The SPX surface contains limited information about the absolute level of the flow-impact coefficient. Nevertheless, retaining $c$ in the Q-measure calibration is necessary for identifying the Q--P wedge interpreted as the market price of flow risk. We therefore report the point estimate together with its profile uncertainty and do not interpret the point estimate as a sharply identified structural coefficient.

\end{document}